\documentclass[11pt]{article}

\usepackage{fullpage}
\usepackage{amsmath}
\usepackage{amssymb}
\usepackage{amsthm}
\usepackage{complexity}

\usepackage{tikz}
\usetikzlibrary{arrows,backgrounds,positioning,calc,shapes.geometric,decorations.pathmorphing}
\usepackage{caption}

\newtheorem{definition}{Definition}
\newtheorem{theorem}{Theorem}

\begin{document}

\title{{\bf Matchgates Revisited}}

\vspace{0.3in}
\author{Jin-Yi Cai\thanks{University of Wisconsin-Madison and Peking University.
 {\tt jyc@cs.wisc.edu}. Supported by NSF CCF-0914969.}
\and Aaron Gorenstein\thanks{University of Wisconsin-Madison.
 {\tt agorenst@cs.wisc.edu}}}


\newcommand{\Pf}{\mathrm{Pf}}
\newcommand{\PfS}{\mathrm{PfS}}
\newcommand{\gp}{Grassmann-Pl\"ucker}
\newcommand{\PerfMatch}{\mathrm{PerfMatch}}
\newcommand{\orientg}{\overrightarrow{G}}
\newcommand{\order}[1]{\left| #1 \right|}

\newcommand{\tensor}{\otimes}

\newcommand{\orientgalpha}[1][\alpha]{\overrightarrow{G^{{#1}}}}

\newcommand{\alphastring}{u b v c w}
\newcommand{\talphastring}{\tilde{u} b \tilde{v} c \tilde{w}}
\newcommand{\bartalphastring}{\tilde{u} \overline{b} \tilde{v} \overline{c} \tilde{w}}
\newcommand{\baralphastring}{u\overline{b}v\overline{c}w}

\date{}
\maketitle

\begin{abstract}
We study a collection of concepts and theorems that laid the foundation
of matchgate computation. This includes the signature theory of
planar matchgates, and the parallel theory of characters of
not necessarily planar matchgates. 
Our aim is to present  a unified and, whenever possible, simplified
account of this challenging theory. Our results include: (1) A direct
proof that Matchgate Identities (MGI) are necessary and sufficient conditions
for matchgate signatures. This proof is self-contained
and does not go through 
the character theory. More importantly it rectifies a gap in the existing 
proof. (2) A proof that Matchgate Identities
 already imply the Parity Condition.
(3) A simplified construction of a crossover gadget. 
This is used in the proof of sufficiency of MGI for
matchgate signatures. This is also used to
give a proof of equivalence between the signature theory and
the character theory which permits omittable nodes.
(4) A direct construction of matchgates realizing
all matchgate-realizable symmetric signatures.
\end{abstract}

\section{Introduction}
Leslie Valiant introduced {\it matchgates} in a seminal paper~\cite{valiant02a}.
In that paper he presented a way to encode computation via 
Pfaffian and Pfaffian Sum, and showed that a non-trivial, though
restricted, fragment of
quantum computation can be simulated in classical polynomial time.
Underlying this magic is a way to encode 
certain quantum states by a classical computation of perfect matchings,
and to simulate certain quantum gates by the so-called matchgates.
These matchgates are weighted graphs, not necessarily planar,
and are equipped with input and output nodes, as well as the so-called
omittable nodes.  Each matchgate is associated with a {\it character},
whose entries are defined in terms of Pfaffian and Pfaffian Sum.

Three years later, there was great excitement
when Valiant invented holographic algorithms~\cite{valiant08},
where he also introduced planar matchgates.
These matchgates are planar graphs, have a subset of
vertices on the outer face designated as external nodes,  and each
matchgate is associated with a
{\it signature}.
The entries of a signature are defined in terms of
the perfect matching polynomial, $\PerfMatch(\cdot)$.
For planar weighted graphs, this quantity
can be computed by the well-known
Kasteleyn's algorithm~\cite{kasteleyn67} (a.k.a. FKT algorithm~\cite{temperleyfisher61})
in polynomial time,
which uses Pfaffian and a Pfaffian orientation.

These holographic algorithms are quite exotic,
and use a quantum-like superposition of fragments of
computation to achieve custom designed cancellations.
The two basic ingredients of holographic algorithms from~\cite{valiant08}
are matchgates and holographic transformations.
A number of concrete problems are shown to be
polynomial time computable by this novel technique, 
even though they appear to require exponential time,
and minor variations  of which   are NP-hard.
They challenge our perceived boundary of 
what polynomial time computation
can do. Since we don't really have any reasonable absolute lower 
bounds that apply to unrestricted computational models,
our faith in such  well-known conjectures
such as P $\not =$ NP or P $\not =$ ${\rm P}^{\rm \#P}$ are based primarily
on the inability of existing algorithmic techniques to solve
NP-hard or \#P-hard problems in polynomial time.
To maintain this faith,
it is imperative that we gain a better understanding of
 what the new methodology
can or cannot do.
To quote Valiant~\cite{valiant08}, ``any proof of P $\not =$ NP 
may need to
explain, and not only to imply, the unsolvability''
of NP-complete or \#P-complete problems by this
methodology.
It becomes apparent that there is a fundamental problem
of what are the intrinsic limitations of
these matchgates, and what is the relationship between
characters of general matchgates and signatures of
planar matchgates.

In~\cite{valiant02}, Valiant showed that the
character of every 2-input 2-output
matchgate must satisfy five polynomial identities, called
{\it Matchgate Identities}.  Valiant used this to show that certain
quantum gate cannot be simulated by the characters of
these matchgates.  
In a sequence of two papers~\cite{caichoudhary07a,caichoudharylu09}
a general study of the character theory and the signature theory
of matchgates was undertaken. These papers achieved the following
general results: Essentially there is an equivalence
between the character theory and the signature theory of matchgates,
and a set of {\it useful Grassmann-Pl\"{u}cker identities}
together with the Parity Condition are a necessary and sufficient
condition for a sequence of values to be the signature of a planar
matchgate.
This set of  useful Grassmann-Pl\"{u}cker identities
will be called Matchgate Identities (MGI) in the general sense.
Along the way they also established a concrete characterization
of symmetric signatures, which are signatures whose entries only depend on
the Hamming weight of the index.

However, this proof is tortuous. In particular the proof 
for the signature theory of planar matchgates goes through
characters.  More importantly, there is a subtle but important
 gap in the proof that every planar matchgate signature
must satisfy these Matchgate Identities.  The gap has to do
with the non-uniform and exponentially many ways in which
the induced 
Pfaffian orientations on subgraphs of a planar graph
can introduce a correction factor $(-1)$ to Pfaffian values, relative to
perfect matchings.

%
%
%


In this paper we present a full, self-contained proof that MGI characterize planar matchgate signatures.
This proof does not involve  character theory or any non-planar matchgate.
Moreover, we include a short proof demonstrating that MGI imply the Parity Condition.
Previously this was presented as a separate requirement for matchgates, but now we show  MGI entirely characterize matchgate signatures.
We then revisit and clarify the equivalence between planar matchgates and the original general matchgates.
Along the way we introduce a concise matchgate for the ``crossover gadget'', using only real weights $1$ and $-1$.
Previously the only known such gadget uses complex values. 
Finally, it has been known that the MGI greatly simplify for symmetric 
signatures. By the general theory
any symmetric sequence satisfying MGI must be realizable as the signature of a 
planar matchgate. However, previously this existence is only known
by going through the entire equivalence proof of characters and
signatures, which also uses the only known ``crossover gadget''.
In  this paper,  we present a simple, direct construction 
of a planar matchgate realizing any symmetric sequence satisfying MGI.

The most intricate part of this paper is the proof that planar matchgate signatures must satisfy  MGI.
The  subtle gap in the existing proof  stems from the following.
To compute the signature of a matchgate $G$, we assume it has a fixed Pfaffian orientation $\orientg$.
This induces a natural Pfaffian orientation for every subgraph, $\orientgalpha$, where $\alpha$ is a bitstring specifying  
a removal pattern of the external nodes from $G$.
A Pfaffian orientation may introduce
 an extra $(-1)$ factor, a ``sign change'', to the corresponding 
perfect matching value.
However, whether $\orientg$ has a sign change does not immediately imply if $\orientgalpha$ has a sign change: the presence or absence of the ``sign change'' may itself change between different external node removals!
As there are exponentially many possible bitstrings $\alpha$
for possible removal patterns, this severely complicates any proof trying to show they must satisfy MGI.
We note that Pfaffian orientations are themselves an important topic \cite{thomas06}, and this result is the first to our knowledge concerning the behavior of 
Pfaffian orientations of exponentially many subgraphs under node removals.

Thus, our main goal is to show that the \emph{change} of the sign change occurs in a pattern such that  MGI still hold.
We do so using Theorem~\ref{thm:delta}.
Essentially, it proves the following.
For any two fixed bit positions $i < j$ referencing the external nodes, 
let $b_i b_j  
\in \{0, 1\}^2$ be the bit pattern on these two bits. 
Then, while the sign change 
may be different for different values of $b_i b_j$,
{\it the change of sign change} when we 
go from $b_i b_j$ to $\overline{b_i}~ \overline{b_j}$ is always
the same, independent of the removal pattern on the other external nodes.
This is succinctly expressed as a {\it quadruple product identity}.
Moreover, 
this is in fact  the strongest statement we can say about a pair of nodes
and their change of signs,  (see Fig.\ \ref{fig:complexdelta}).
Fortunately this is also sufficient to prove  MGI.

This paper is organized as follows.
In  Section~\ref{Preliminaries} we define all the concepts
and terminology in the signature theory of planar matchgates. We will also prove
that MGI imply the Parity Condition. We will restrict to planar
matchgates pertaining to signature theory here. The terminology
having to do with general (not necessarily planar) matchgates and
characters will be delayed till Section~\ref{Character}.
In Section~\ref{Pfaffian-Signature-Identities}
we will give a self-contained proof of some
known identities.  This is partly for the convenience of the readers,
partly to give simplified proofs when possible. For example,
the earlier proof of Theorem~\ref{gpi2} from \cite{dresswenzel95}
goes through skew-symmetric bilinear forms and operators acting 
on the exterior
algebra of a module over some commutative ring.
Here we present a direct, elementary proof.
In Section~\ref{Matchgates-sat-MGI}
we prove that every matchgate signature satisfies MGI.
In Section~\ref{sec:realizing} we prove that
MGI are also sufficient to be realizable as a matchgate signature.
Here we also give the simplified construction of a crossover gadget.
In Section~\ref{Character} we discuss the character theory.
In Section~\ref{Sym-Signatures} we give the direct  construction for 
matchgates realizing symmetric
signatures.
Some concluding remarks are in Section~\ref{Conclusion}.

\section{Preliminaries}\label{Preliminaries}

\paragraph{Matchgate, PerfMatch definitions}
A matchgate is an undirected weighted plane graph $G$ with $k$ distinguished ``external'' nodes on its outer face, ordered in a clockwise order.
(We will see shortly that without loss of generality we may assume
the graph $G$  is connected.
Therefore it is a plane graph, i.e., a planar graph given with a
 particular planar embedding, and the outer face is both uniquely defined and has a connected boundary.)
Without loss of generality,
we assume all edge weights are non-zero; zero weighted edges
can be deleted.
We define the perfect matching polynomial, $\PerfMatch(G)$, as the following:
\begin{equation}\label{eq:perfmatch}
\PerfMatch(G)=\sum_{M\in\mathcal M(G)}\prod_{e\in M}w(e)
\end{equation}
where $\mathcal M(G)$ is the set of all perfect matchings in $G$ and $w(e)$ is the weight of edge $e$ in $G$.
For each length-$k$ bitstring $\alpha$, $G$ defines a subgraph $G^\alpha$
obtained from $G$ by the following operation:
For  all $1 \le i \le k$, if the $i$-th bit $\alpha_i$
of $\alpha$ is $1$, then we remove the $i$-th external node and 
all its incident edges.
Thus, $G^{00\ldots0}=G$, and $G^{11\ldots1}$ is $G$ with \emph{all} external nodes removed.

\paragraph{Signature, perfect matching term definitions}
We define the signature of the matchgate $G$
 as the  vector $\Gamma_G = (\Gamma_G^\alpha)$, indexed by $\alpha \in \{0, 1\}^k$, as follows:
\begin{equation}
\Gamma_G^\alpha=\PerfMatch(G^\alpha)=\sum_{M\in\mathcal M(G^\alpha)}\prod_{e\in M}w(e).
\end{equation}
For a perfect matching $M\in \mathcal M(G^\alpha)$ 
we define $\Gamma_G^\alpha(M) = \prod_{e\in M}w(e)$ as the \emph{perfect matching term}, equal to the product of the edge weights for the matching $M$.
Where $G$ is clear, we omit the subscript $G$, and 
write $\Gamma^\alpha$ for $\Gamma_G^\alpha$,
and $\Gamma^\alpha(M)$ for $\Gamma_G^\alpha(M)$.

\paragraph{Pfaffian orientations, induced Pfaffian orientations}
For a plane graph $G$, we can compute $\PerfMatch(G)$ using  
Kasteleyn's algorithm~\cite{kasteleyn67} via the Pfaffian. 
A \emph{Pfaffian orientation} on $G$
is an assignment of a direction to each edge of $G$  in  such a way that
each face, except possibly the outer face, has an odd number of 
clockwise oriented edges when one traverses the boundary of the
face.
Such an orientation is easy to compute for any plane graph.
Note that any ``bridge edge'' (an edge both sides of which belong
to the same face) can be oriented arbitrarily, and the traversal
of the face will count the edge twice, once clockwise and once
counter-clockwise. 
Under a \emph{Pfaffian orientation} on $G$,
the Pfaffian of a  skew-symmetric  matrix defined by $G$
and the orientation, defined below, is equal to $\pm \PerfMatch(G)$.
We fix a single Pfaffian orientation for $G$ 
and call the directed graph $\orientg$.
Note that $\orientg^\alpha$, which is obtained from $\orientg$
by removing some vertices and  their incident edges according to $\alpha$, 
is \emph{also} Pfaffian-oriented.
This is because we only remove zero or more vertices on the outer face, and the removal of these vertices and their incident edges do not create any \emph{non-outer face}.
Thus a single fixed Pfaffian
orientation for $G$ induces a set of Pfaffian 
orientations, one for each $G^\alpha$.
We consider a Pfaffian orientation for $G$ is fixed,
and each $G^\alpha$ inherits the induced Pfaffian  orientation.

\paragraph{Skew-symmetric matrix}
Now we assume the vertices of $G$ are labeled by a totally
ordered set, for example, $1 < 2 < \ldots < n$.
Given an orientation on $G$, we define a 
skew-symmetric adjacency matrix $A =A_{\orientg}$  for $\orientg$
as follows.
Let $(u,v)$  be a directed edge from $u$ to $v$ in   $\orientg$.
Then $A_{u,v}=w(\{u,v\})$, and $A_{v,u}=-w(\{u,v\})$,
where $w(\{u,v\})$ is the weight of the corresponding edge in $G$.
Note that if the labels $u<v$, then the entry above the diagonal $A_{u,v}=w(\{u,v\})$, and its reflected entry below the diagonal $A_{v,u}=-w(\{u,v\})$.
If $u>v$ then the entry above the diagonal $A_{v,u}=-w(\{u,v\})$ and its reflected entry below the diagonal $A_{u,v}=w(\{u,v\})$ instead.
The diagonal and all other locations $(u,v)$
not corresponding to an edge in the matrix $A$ are set to $0$.
The lower-left triangle of $A$ is the negation of the upper-right triangle.

\paragraph{Pfaffian}
The Pfaffian of an $n\times n$ matrix, where $n \ge 2$ is even, 
is defined as follows:
\begin{equation}\label{Pfaffian-defined}
\Pf(A)=\sum_\pi\epsilon_\pi A_{i_1,i_2}A_{i_3,i_4},\ldots,A_{i_{n-1},i_{n}}
\end{equation}
where the sum is over all permutations $\pi =
{\scriptsize \begin{pmatrix} 1 & 2 & \ldots & n \\
i_1 & i_2 &  \ldots & i_n \end{pmatrix}}$ 
such that $i_1<i_2$, $i_3<i_4$, $\ldots$, $i_{n-1}<i_n$ 
and $i_1<i_3<i_5<\ldots<i_{n-1}$.
The term $\epsilon_\pi$ is $-1$ or $1$ depending on 
whether the parity of $\pi$ is 
odd or even, respectively.
We note that there is a natural 1-1 correspondence between
permutations $\pi$ in this \emph{canonical} expression
and the set of partitions of $[n]$ into disjoint pairs, which
are potential perfect matchings.
A permutation  $\pi$ corresponds to an actual perfect matching iff
all the pairs are edges.
It is known and easy to verify that the sign $\epsilon_\pi$
can also be computed by the parity of the number of overlapping pairs
($+1$ if it is even, $-1$ if it is odd).
We say $\{i_{2k-1},i_{2k}\}$ and $\{i_{2\ell-1},i_{2\ell}\}$ is
an overlapping pair  iff $i_{2k-1}<i_{2\ell-1}<i_{2k}<i_{2\ell}$ or  $i_{2\ell-1}<i_{2k-1}<i_{2\ell}<i_{2k}$.

We note that the 
term $\epsilon_\pi A_{i_1,i_2}A_{i_3,i_4},\ldots,A_{i_{n-1},i_{n}}$
is the same for any listing of the partition
$[n] = \{i_1, i_2\} \cup \{i_3, i_4\}  \cup \ldots \cup \{i_{n-1}, i_n \}$,
where $\pi =
{\scriptsize \begin{pmatrix} 1 & 2 & \ldots & n \\
i_1 & i_2 &  \ldots & i_n \end{pmatrix}}$,
independent of the ordering of the pairs, as well as the order within 
each pair. We also note that this definition is valid for any
linear order on the vertices; it need not be the set of consecutive
integers from 1 to $n$.  This is particularly relevant when
we consider the Pfaffian of  $\orientg^\alpha$, where
the vertices will inherit the labeling from $G$.

As convention, 
if $n$ is odd, then $\Pf(A)=0$;
if $n$ is zero, then $\Pf(A)=1$.

\paragraph{Relating $\Pf$ to $\PerfMatch$}
If $A=A_{\orientg}$,
we call $\epsilon_\pi A_{i_1,i_2}A_{i_3,i_4},\ldots,A_{i_{n-1},i_{n}}$ a 
\emph{Pfaffian term}.
As observed, there is a 1-to-1 correspondence
between all  non-zero Pfaffian terms and perfect matchings
in $\mathcal M(G)$. 
 If
$M$ is a perfect matching, we denote the corresponding Pfaffian term by
$\Pf_{\orientg}(M)$.
A  perfect matching term has the same value, up to a $\pm$ sign, 
as the corresponding Pfaffian term.
In other words, $\Pf_{\orientg}(M) = \pm\Gamma_G(M)$.
They may indeed differ, even under a Pfaffian orientation.
The heart of the FKT algorithm is the proof that for
the skew symmetric matrix of a 
Pfaffian-oriented graph, either every pair
of corresponding terms are the same, or
every pair of corresponding terms differ by  a minus sign.
Thus,
$\Pf(A_{\orientg})=\pm\PerfMatch(G)$.
This equality is an equality of  polynomials:  
Given a Pfaffian oriented $\orientg$,
there exists an $\epsilon = \pm 1$,
such that
\begin{equation}\label{pftoperfmatch}
\Pf(A_{\orientg})= \epsilon \PerfMatch(G)
\end{equation}
and if (\ref{pftoperfmatch}) holds
for one set of edge weights,
then \emph{every} Pfaffian term 
is $\epsilon$ times its corresponding perfect matching term, for every 
set of weights.

\paragraph{Pfaffian signature definition}
As the orientation in $\orientg$ induces a Pfaffian orientation for all $G^\alpha$, we can naturally refer to $\orientgalpha$.
Note that
$\orientgalpha = \orientg^\alpha$, the oriented graph obtained by
$\orientg$ after removing some vertices and incident edges
according to $\alpha$, in the same way as before.
Also note that
 $A_{\orientgalpha}$ is obtained from 
$A_{\orientg}$ by removing the appropriate columns and rows
indicated by $\alpha$.
We abbreviate $\Pf(A_{\orientgalpha})$ as $\Pf_{\orientg}^\alpha$.
Where $\orientg$ is clear, we just write $\Pf^\alpha$.
With a given Pfaffian orientation on the plane graph $G$, 
and a given labeling of its $k$ external nodes
in  clockwise order, we define the Pfaffian Signature of $\orientg$ to be the vector $(\Pf^\alpha)$ indexed by $\alpha\in\{0,1\}^k$.
Each $\Pf^\alpha$ is a sum of Pfaffian terms, by the definition
of $\Pf(A_{\orientgalpha})$,
under the induced Pfaffian orientation.

Critically, eq.\ (\ref{pftoperfmatch}) is a term by term equation:
For every $\alpha \in \{0, 1\}^k$,
there exists $\epsilon(\alpha) \in \{-1, 1\}$, such that for all $M
\in{\mathcal M}(G^\alpha)$,
\begin{equation}
\Pf_{\orientgalpha}(M)= \epsilon(\alpha) \Gamma_{G^\alpha}(M).
\end{equation}

\paragraph{Matchgate Identities}
We state the Matchgate Identities, or MGI.
\begin{theorem}\label{thm:mgi}
Let $\Gamma$ be the signature of a matchgate with $k$ external nodes.
For any length-$k$ bitstrings  $\alpha,\beta \in \{0, 1\}^{k}$, let
$\alpha\oplus\beta \in \{0, 1\}^{k}$ be
their bitwise XOR, and let $P=\{p_1,\ldots,p_l\}$, where $p_1<\ldots<p_l$,
be the subset of $[k]$ whose characteristic sequence is $\alpha\oplus\beta$.
Here $p_i$ is the $i$-th bit where $\alpha$ and $\beta$ differ.
Then,
the signature $\Gamma$  satisfies: 
\begin{equation}\label{eq:mgi}
\sum_{i=1}^{l}(-1)^{i}\Gamma^{\alpha\oplus e_{p_i}}\Gamma^{\beta\oplus e_{p_i}}=0,
\end{equation}
where $e_j$ denotes a length-$k$ bitstring with a $1$ in the $j$-th index, and $0$ elsewhere.
\end{theorem}
We will show that this is a complete characterization of what vectors
 can be planar matchgate signatures.

\paragraph{Parity Condition}
A perfect matching has an even number of vertices. Therefore it follows that
$\PerfMatch(G^\alpha)=0$, whenever $G^\alpha$ has an odd number of vertices.
Thus, either for all $\alpha$ of odd Hamming weight,
or for all $\alpha$ of even  Hamming weight,
 $\Gamma^\alpha=0$.

\paragraph{Matchgate Identities Imply Parity Condition}
Here we show that this Parity Condition is  a consequence 
of MGI.
\begin{theorem}\label{MGI-implies-Parity}
If a vector $\Gamma$ obeys the MGI, then it also obeys the Parity Condition.
\end{theorem}
\begin{proof}
For a contradiction assume $\Gamma^\alpha\neq0$
and $\Gamma^\beta\neq0$, for some $\alpha$ and
$\beta$
of  even and odd Hamming weight respectively.
We define $\widetilde{\Gamma}$ 
by $\widetilde{\Gamma}^\gamma=\Gamma^{\gamma\oplus{\alpha}}$.
Since $\gamma \oplus \gamma' = (\gamma\oplus{\alpha}) \oplus
(\gamma' \oplus{\alpha})$, the vector $\Gamma$ obeys the MGI
implies that the vector $\widetilde{\Gamma}$ also obeys the MGI.
Also $\widetilde{\Gamma}^{00\ldots0} = \Gamma^{\alpha}  \neq0$
and $\widetilde{\Gamma}^{\beta \oplus\alpha}  = \Gamma^{\beta}  \neq0$.
Note that $\beta \oplus\alpha$ has an odd Hamming weight.

Let $\beta' = \{p_1, \ldots, p_l\}$ be of minimum odd Hamming weight such that 
$\widetilde{\Gamma}^{\beta'}\neq0$, where $l \ge 1$.
Now invoke the MGI on the bitstrings $00\ldots0\oplus e_{p_1}$ 
and $\beta'\oplus e_{p_1}$.
That gives
\begin{equation}\label{eq:parity}
0 = -\widetilde{\Gamma}^{00\ldots0}\widetilde{\Gamma}^{\beta'} 
+ \sum_{i=2}^{l}(-1)^i\widetilde{\Gamma}^{00\ldots0\oplus e_{p_1}\oplus e_{p_i}}\widetilde{\Gamma}^{\beta'\oplus e_{p_1}\oplus e_{p_i}}.
\end{equation}
If $l =1$ then the sum $\sum_{i=2}^{l}$ is vacuous,
and we have a contradiction.
So $l \ge 2$ and we consider each term in the sum  $\sum_{i=2}^{l}$.
Observe that for every $2 \le i \le l$,
 $\beta'\oplus e_{p_1}\oplus e_{p_i}$ has an odd 
Hamming weight \emph{less} than that of $\beta'$, 
hence $\widetilde{\Gamma}^{\beta'\oplus e_{p_1}\oplus e_{p_i}}
=0$. 
Thus the sum  $\sum_{i=2}^{l}$ is zero
but $\widetilde{\Gamma}^{00\ldots0}\widetilde{\Gamma}^{\beta'} 
\not =0$, a contradiction.
\end{proof}

Nonetheless, in further development of the signature theory,
our experience is that the Parity Condition
 is a good criterion to apply first.

\paragraph{The Sign}
MGI were first introduced by Valiant in \cite{valiant02}
in the context of proving certain 2-input 2-output quantum gate cannot
be realized by a matchgate.  It was shown that 2-input 2-output matchgates
must satisfy certain identities which are named Matchgate Identities.
These identities are actually concerned with \emph{characters} of
matchgates. These so-called characters are defined directly
in terms of Pfaffians, and their underlying matchgates need not be
planar by definition.  In the case of 2-input 2-output matchgates, these
character values constitute a 4 by 4 matrix, 
called a character matrix.  Subsequently in \cite{caichoudhary07a}
and \cite{caichoudharylu09}, this theory is generalized to
matchgates of an arbitrary number of external nodes. 
The ultimate result is that 
there is an equivalence of matchgate characters
(of not necessarily planar matchgates) and matchgate signatures
(of planar matchgates). See Section~\ref{Character}. Furthermore
Matchgate Identities (together with the Parity Condition) are 
a necessary and sufficient condition for a vector of values to
be the signature of a (planar) matchgate.
By Theorem~\ref{MGI-implies-Parity}, in fact the Matchgate Identities 
already logically imply the Parity Condition.

The existing proof of the equivalence  of
being a matchgate signature and satisfaction of MGI (together with 
parity requirements) is quite long and tortuous.
In particular it goes through characters.
More importantly, there is a gap in the 
existing proof that 
Matchgate Identities are a necessary condition for a matchgate signature.
The gap is to exactly account for the change of signs from Pfaffians to signatures.
We will rectify this situation.
Our new proof is direct and self-contained; we show that Matchgate Identities are a necessary condition for matchgate signatures without going through characters.

We will first establish the Pfaffian Signature Identities.
\begin{theorem}\label{thm:pfafmgi}
Let $\orientg$ be a plane graph with a Pfaffian orientation and
$k$ external nodes.
For any length-$k$ bitstrings  $\alpha,\beta \in \{0, 1\}^{k}$, let
$\alpha\oplus\beta \in \{0, 1\}^{k}$ be
their bitwise XOR, and let $P=\{p_1,\ldots,p_l\}$, where $p_1<\ldots<p_l$,
be the subset of $[k]$ whose characteristic sequence is $\alpha\oplus\beta$.
Then,
\begin{equation}\label{eq:pfafmgi}
\sum_{i=1}^l (-1)^i \Pf^{\alpha\oplus e_{p_i}}\Pf^{\beta\oplus e_{p_i}}=0.
\end{equation}
\end{theorem}
Because of the ``sign change'' between $\Pf^\alpha$ and $\Gamma^\alpha$, this statement does not immediately imply Theorem~\ref{thm:mgi}.
We need to know that the extra $-1$ factors between $\Pf^\alpha$ and $\Gamma^\alpha$ appear in just such a pattern that the $-1$ factors all cancel each other in the Matchgate Identities in (\ref{eq:mgi}) {\it relative to}
 the Pfaffian Signature Identities in (\ref{eq:pfafmgi}).
Before doing so, we will prove Theorem~\ref{thm:pfafmgi}.

\section{Proving the Pfaffian Signature Identities}\label{Pfaffian-Signature-Identities}
Theorem~\ref{thm:pfafmgi} will follow from the {\gp} Identities over Pfaffian minors of a matrix.
We state the following definition of the {\gp} Identities for a skew-symmetric matrix $A$.
In writing $\Pf(i_1,i_2,\ldots,i_L)$ we mean the Pfaffian of the
$L\times L$ matrix whose rows and columns are the $i_1,i_2,\ldots,i_L$-th rows and columns of $A$, in that order.
The order matters: $\Pf(i_1,i_2,\ldots) = -\Pf(i_2,i_1,\ldots)$, for instance.
When we write $\Pf(i_1, i_2, \ldots, \widehat{i_k}, \ldots, i_K)$,
 the $\widehat{i_k}$  means that $i_k$ is explicitly excluded from that list.
\begin{theorem}[The {\gp} Identities]\label{gpi1}
Let $I=\{i_1, i_2, \ldots, i_L\}, J=\{j_1, j_2, \ldots, j_K\}$ be subsets of indices of $A$, where $i_1< i_2 < \ldots < i_L$ and $j_1 < j_2 < \ldots < j_K$.
Then
{\small \begin{equation}\label{eq:gpi1}
\sum_{\ell=1}^L (-1)^{\ell-1}
\Pf(j_\ell, i_1, \ldots, i_K)
\Pf(j_1, \ldots, \widehat{j_\ell}, \ldots, j_L)
 + \sum_{k = 1}^K (-1)^{k-1}
\Pf(i_1, \ldots, \widehat{i_k}, \ldots, i_K)
\Pf(i_k, j_1,  \ldots, j_L)
= 0
\end{equation} }
\end{theorem}

Theorem \ref{gpi1} has the following short proof \cite{valiant02, murota00} originally from \cite{ohta92}.
\begin{proof}[Proof of Theorem \ref{gpi1}]
From the definition of Pfaffian:
\begin{align}
& \Pf(j_\ell, i_1,  \ldots, i_K) = \sum_{k=1}^K (-1)^{k-1} \Pf(j_\ell,i_k)   \Pf(i_1,\ldots,\widehat{i_k},\ldots,i_K)\\
&\Pf(i_k, j_1, \ldots, j_L)  = \sum_{\ell=1}^L (-1)^{\ell-1} \Pf(i_k,j_\ell)\Pf(j_1,\ldots,\widehat{j_\ell},\ldots,j_L)\\
\intertext{and also}
&\Pf(j_\ell,i_k)+\Pf(i_k,j_\ell)=0.
\end{align}
The proof is completed by substituting these into the left hand side of eq.\ (\ref{eq:gpi1}).
\end{proof}
There is another form of these identities which is more closely related to the Pfaffian Signature Identities.
We state this theorem next.
An earlier proof of Theorem~\ref{gpi2} appears in \cite{dresswenzel95}.
They go through skew-symmetric bilinear forms and operators acting on the exterior algebra $\Lambda(M)$ of an $R$-module $M$ over some commutative ring $R$.
Here we present a direct, elementary proof.

\begin{theorem}\label{gpi2}
Let $A,I,J$ be as above.
For a subset $S$ of indices of $A$, we write $\Pf(S)$ when $S$ is listed in increasing order.
Let $D=I \Delta J=\{k_1,\ldots,k_m\}$
(listed in increasing order) be the symmetric difference of $I,J$. Then
\begin{equation}\label{eq:gpi2}
\sum_{s=1}^m (-1)^{s-1}
\Pf(I \Delta \{k_s\})
\Pf(J \Delta \{k_s\})
=0
\end{equation}
\end{theorem}

\begin{proof}[Proof of Theorem~\ref{gpi2}]
We prove Theorem~\ref{gpi2} by Theorem~\ref{gpi1}.

Considering a term in eq.\ (\ref{eq:gpi1}), and
let $x$ be the element being moved from the index set of 
one Pfaffian to another.
If $x \in I \cap J$, clearly the term is 0.
It follows that there is a one-to-one correspondence between the remaining terms in eq.\ (\ref{eq:gpi1}) and (\ref{eq:gpi2}).
All that remains is showing that each such term in eq.\ (\ref{eq:gpi1}) has the same sign as its counterpart in (\ref{eq:gpi2}).

Suppose $x \in J-I$.
In that case, the term in eq.\ (\ref{eq:gpi1}) is
\begin{equation}
(-1)^z \Pf(x, i_1,  \ldots, i_K)\Pf(j_1 \ldots, \hat{x}, \ldots, j_L)
\end{equation}
where $z$ is the number of elements in $J$ preceding $x$, equivalently those elements in $J$ less than $x$.
We write $z=a+b$, where
\begin{align}
a &= \order{\{ y \mid y \in J-I, y < z\}}\\
b &= \order{\{ y \mid y \in J\cap I, y < z\}}\\
\intertext{and we also define}
c &= \order{\{y \mid y\in I-J, y < z\}}.
\end{align}
When we put the indices in $\Pf(x,i_1,\ldots,i_K)$ in increasing order we move $x$ along until it is in the sorted order, we move $x$ exactly $b+c$ times.
Thus
\begin{equation}
\Pf(x,i_1,\ldots i_K) = (-1)^{b+c}\Pf(I\cup\{x\})
\end{equation}
and so it follows that
\begin{equation}
(-1)^z \Pf(x,i_1,\ldots i_K) = (-1)^{a+c} \Pf(I\cup \{x\}).
\end{equation}
It is clear that $a+c$ is precisely the number of those in $D$ preceding $x$, exactly the sign in front of the corresponding term in (\ref{eq:gpi2}).

The argument for the case $x \in J-I$ is symmetric.
\end{proof}

Now we are ready to prove Theorem~\ref{thm:pfafmgi}.
\begin{proof}[Proof of Theorem~\ref{thm:pfafmgi}]
We prove  Theorem~\ref{thm:pfafmgi} by Theorem~\ref{gpi2}.
For a matchgate $G$, let $\alpha,\beta$ be two bitstrings of length $k$, where $k$ is the number of external nodes in $G$.
The $i$-th bit of $\alpha$, denoted $\alpha_i$, corresponds to the $i$-th external node in $G$ in clockwise order.

Let $U$ be the set of all internal (that is, not external) nodes in $G$.
We define $I=\{v_i \mid \alpha_i=0\}\cup U$,
where $v_i$ is the label of the vertex in $G$ which is the $i$-th
external node referenced by $\alpha_i$. Similarly let
$J=\{ v_i \mid \beta_i = 0 \} \cup U$.
Observe that $I \Delta J= \{v_i \mid \alpha_i \neq \beta_i\}$.
It follows that there is a term-for-term correspondence between 
(\ref{eq:gpi2}) of Theorem~\ref{gpi2} and (\ref{eq:pfafmgi}) of Theorem~\ref{thm:pfafmgi}.
\end{proof}

\section{Matchgates Satisfy Matchgate Identities}\label{Matchgates-sat-MGI}
We will now prove that while $\Pf^\alpha$ may differ from $\Gamma^\alpha$ by a sign depending on $\alpha$, the differences occur in just such a pattern that they cancel in the MGI.
This will allow us to conclude that the Pfaffian Signature Identities (\ref{eq:pfafmgi}) differ from the MGI (\ref{eq:mgi}) by a global $\pm 1$ factor, thus 
proving  the Matchgate Identities.

\begin{definition}
For any $M\in\mathcal M(G^\alpha)$, where $G^\alpha$ has the
orientation $\orientgalpha$, we define the \emph{sign} of 
the perfect matching $M$ to be:
\begin{equation}
{\rm sgn}(M) = \frac{\Pf_{\orientgalpha}(M)}{\Gamma_{G^\alpha}(M)}\in\{-1,1\}
\end{equation}
\end{definition}
Recall that  it is a polynomial equality that
the Pfaffian is equal to
 $\pm\PerfMatch$, under a Pfaffian orientation.
Thus we can conclude that, for 
$\Pf_{\orientgalpha}(M)$ and $\Gamma_{G^\alpha}(M)$,
the value of ${\rm sgn}(M)$ is the same $\pm 1$ for every perfect matching 
$M\in\mathcal M(G^\alpha)$.
This allows us to define a very useful function:
\begin{definition}
For any $\alpha$ such that $\mathcal M(G^\alpha)\neq\emptyset$, we take
any $M\in\mathcal M(G^\alpha)$ and define the function $\delta$:
\begin{equation}
\delta(\alpha)=
{\rm sgn}(M) = \frac{\Pf_{\orientgalpha}(M)}{\Gamma_{G^\alpha}(M)}
\end{equation}
\end{definition}
Note that $\delta(\alpha)$ is well-defined;
the value is independent of the choice of $M\in\mathcal M(G^\alpha)$.
It is defined whenever $\mathcal M(G^\alpha)\neq\emptyset$.
Recall that we have a fixed Pfaffian orientation for $G$ and a fixed
induced orientation for all $G^\alpha$.

We are ready to state the key theorem which implies the MGI.
\begin{theorem}\label{thm:delta}
Let $ubvcw \in \{0,1 \}^k$, where
$u\in\{0,1\}^{i-1}$, $b$ refers to the $i$-th bit,
$v\in\{0,1\}^{j-i-1}$, $c$ refers to the $j$-th bit,
and $w\in\{0,1\}^{k-j}$, with $1 \le i < j \le k$.
Let $\talphastring \in \{0,1 \}^k$ be a possibly different bitstring, 
but with $b,c$ still referencing the $i$-th bit and $j$-th bit,
respectively.
Let $\overline{b}=1-b$ and $\overline{c}=1-c$.
Then the following is true:
\begin{equation}\label{eq:deltaequality}
\delta(\alphastring)\delta(\baralphastring) = \delta(\talphastring)\delta(\bartalphastring)
\end{equation}
when all four $\delta$ terms involved are defined.
\end{theorem} 
Note that the only equality we claim here is the pair-wise product being
the same. 
The individual $\delta$ terms can vary; 
for example
there are cases when the above equation resolves to $(1)(-1)=(-1)(1)$ (see Fig.\ \ref{fig:complexdelta}).
The theorem asserts that if flipping two fixed bits changes the ``sign
change'' $\delta$
for some $u,v,w$, then it will 
change the sign change for all $u,v,w$ of the same lengths
 whenever $\delta$ is defined.
It is an invariance of {\it the change of sign change}. 

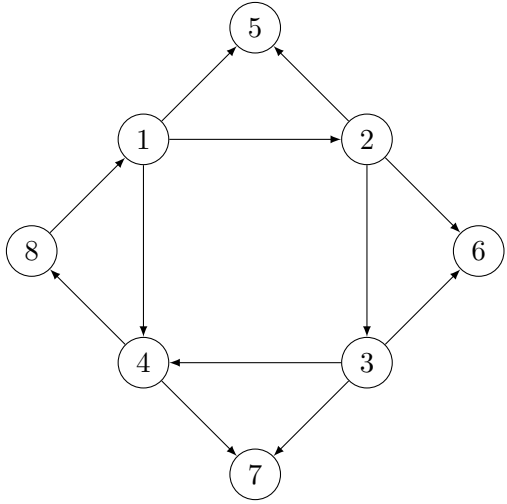
\begin{figure}
\centering
\begin{tikzpicture}
[every path/.style={>=latex},vertex/.style={circle,draw},input/.style={circle,draw}]
\node[vertex](1){1};
\node[vertex](5)[above right=of 1]{5};
\node[vertex](2)[below right=of 5]{2};
\node[vertex](6)[below right=of 2]{6};
\node[vertex](3)[below left=of 6]{3};
\node[vertex](7)[below left=of 3]{7};
\node[vertex](4)[above left=of 7]{4};
\node[vertex](8)[above left=of 4]{8};

\foreach \p/\q in {1/2,2/3,3/4,1/4,  1/5,2/5,   2/6,3/6,   3/7,4/7,    4/8,8/1} {
\draw[->] (\p) to (\q);
}
\end{tikzpicture}
\caption{An example of a nontrivial instance of equation (\ref{eq:deltaequality}).
Let the external nodes be $5,6,7,8$.
Observe that $\delta(0000)=1,\delta(1100)=-1,\delta(0011)=-1,\delta(1111)=1$.
Thus, if we let $b,c$ refer to the first two external nodes, $u,v$ both be the empty string, and $w=00$ and $\tilde{w}=11$, we get the situation where equation (\ref{eq:deltaequality}) becomes  $(1)(-1)=(-1)(1)$.}\label{fig:complexdelta}
\end{figure}

Since each factor in  (\ref{eq:deltaequality})
is $\pm 1$, this equation can also be equivalently 
expressed as the following {\it quadruple product identity}:
\begin{equation}\label{eq:quadruple-product-deltaequality}
\delta(\alphastring)\delta(\baralphastring)  \delta(\talphastring)\delta(\bartalphastring) =1
\end{equation}

\paragraph{Theorem~\ref{thm:delta} implies the MGI}
Before proving Theorem~\ref{thm:delta} we show how it proves Theorem~\ref{thm:mgi}.

If there are no non-zero terms in a particular MGI indexed by $\alpha,
\beta \in \{0, 1\}^k$, then the MGI is trivial.

There is a 1-1 correspondence between the non-zero terms in
(\ref{eq:mgi}) and (\ref{eq:pfafmgi}).
Since (\ref{eq:pfafmgi}) is an equality, if there are non-zero terms
in (\ref{eq:mgi}), then there are at least two such terms.
Consider all non-zero terms in  (\ref{eq:mgi}),
and let each  non-zero term from the Pfaffian identity 
(\ref{eq:pfafmgi}) be divided by its corresponding  MGI term in (\ref{eq:mgi}).
The ratio is of the form
\begin{equation}
\frac{\Pf^{\alpha\oplus e_i}\Pf^{\beta\oplus e_i}}
{\Gamma^{\alpha\oplus e_i}\Gamma^{\beta\oplus e_i}} 
=
\delta(\alpha\oplus e_i)\delta(\beta\oplus e_i),
\end{equation}
where $i$ is a bit location where $\alpha_i \not = \beta_i$.
Consider any two such terms and form the \emph{product} of the two products
 of the pairs. 
This quadruple product has the form
\begin{equation}\label{4-way-product}
\delta(\alpha\oplus e_i)\delta(\beta\oplus e_i)
\delta(\alpha\oplus e_j)\delta(\beta\oplus e_j)
\end{equation}
for some $1 \le i < j \le k$, which is the same as 
\begin{equation}\label{reshuffle}
\delta(\alpha\oplus e_i)\delta(\alpha\oplus e_j)
\delta(\beta\oplus e_j)\delta(\beta\oplus e_i).
\end{equation}
Let $\alpha_\ell$ be the $\ell$-th bit in $\alpha$ ($1 \le \ell \le k$).
Recall that $\alpha_i = \overline{\beta_i},\alpha_j = \overline{\beta_j}$.
Letting $b=\overline{\alpha_i}$ and $c={\alpha_j}$, we see that we can use Theorem~\ref{thm:delta} to conclude that
 the product of the first two terms equals the product of the other two terms in (\ref{reshuffle}), 
and so the whole product must be $1$.
This implies that all Pfaffian identity terms differ from their corresponding MGI terms by the same global $\pm1$ constant.
Note that $\delta$ is not defined exactly when that term in the MGI is 0
(and the corresponding Pfaffian Signature Identity term is also 0), so it is sufficient
to consider only those terms in the MGI where the relevant
$\delta$ is defined.

Theorem~\ref{thm:mgi} is proved assuming Theorem~\ref{thm:delta}.

\vspace{.2in}
Now we will prove Theorem~\ref{thm:delta}.
We first prove for the case  $b=c=0$.
The proof for the case $b=1,c=0$ is similar with only a few extra complications.
The other cases follow by symmetry.

\paragraph{Preprocessing} We assume that $G$ is preprocessed in the following 
way:
First we  append a path of length 2 from each external node in $G$.
For the $i$-th external node, we will connect it to a 
new node called $\hat{i}$, 
which is then connected to another new node called $i$.
The new nodes $1, 2, \ldots, k$ are now considered external nodes,
and are labeled as such within the graph. 
All other nodes, including all original nodes and all 
$\hat{1}, \hat{2}, \ldots, \hat{k}$ are non-external nodes.
$\hat{i}$ will be given the label $2k +1 -i$. Thus 
$\hat{1}, \hat{2}, \ldots, \hat{k}$ are ordered 
reversely $2k > 2k-1 > \ldots > k+1$ respectively.  
All other nodes (the original nodes of $G$)
are labeled arbitrarily starting  from $2k+1$.
The modified graph will now be called $G$.
Now all external nodes are at the end of a path of length at least $2$.
It is easy to check that the signature $\Gamma$ is not changed.
As an example of the preprocessing for $k=5$, consider Fig.\ \ref{fig:preprocessing}.

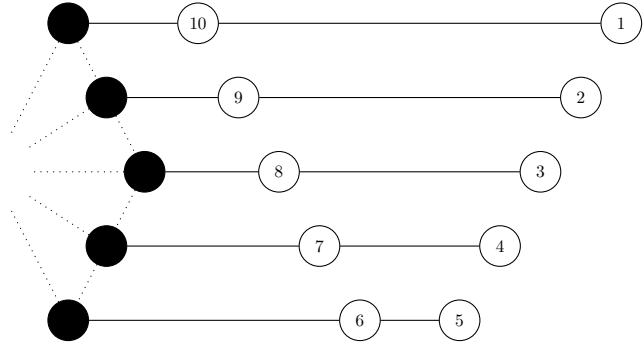
\begin{figure}
\centering
\begin{tikzpicture}
[vertex/.style={circle,draw,scale=1,minimum size=.9cm},input/.style={circle,draw},scale=0.6,transform shape]
\node[circle](scaffold1){};
\node[circle](scaffold2)[right=2.5cm of scaffold1]{};
\node[circle](scaffold3)[right=9cm of scaffold2]{};

\node[vertex,black,fill](1)[below=0cm of scaffold1]{};
\node[vertex,black,fill](2)[below right=1cm and 0.2cm of 1]{};
\node[vertex,black,fill](3)[below right=1cm and 0.2cm of 2]{};
\node[vertex,black,fill](4)[below left= 1cm and 0.2cm of 3]{};
\node[vertex,black,fill](5)[below left= 1cm and 0.2cm of 4]{};
\node[circle,scale=5](core)[left=2cm of 3]{};

\node[vertex](a1)[below =0cm  of scaffold2]{$10$};
\node[vertex](a2)[below right=1cm and 0.25cm of a1]{$9$};
\node[vertex](a3)[below right=1cm and 0.25cm of a2]{$8$};
\node[vertex](a4)[below right=1cm and 0.25cm of a3]{$7$};
\node[vertex](a5)[below right=1cm and 0.25cm of a4]{$6$};

\node[vertex](b1)[below=0cm of scaffold3]{$1$};
\node[vertex](b2)[below left=1cm and 0.25cm of b1]{$2$};
\node[vertex](b3)[below left=1cm and 0.25cm of b2]{$3$};
\node[vertex](b4)[below left=1cm and 0.25cm of b3]{$4$};
\node[vertex](b5)[below left=1cm and 0.25cm of b4]{$5$};

\foreach \p in {1,2,3,4,5} {
\draw [-] (\p) to (a\p);
}
\foreach \p in {1,2,3,4,5} {
\draw [-] (a\p) to (b\p);
}
\foreach \p/\q in {1/2,2/3,3/4,4/5} {
\draw [dotted] (\p) to (\q);
}
\foreach \p in {1,2,3,4,5} {
\draw [dotted] (\p) to (core);
}
\end{tikzpicture}
\caption{An example of preprocessing.}\label{fig:preprocessing}
\end{figure}

Second,  we make $G$ a   connected graph.
If the graph is already connected then we do nothing.
Suppose it is not connected and there are  several connected components $G_i$.
Consider a clockwise traversal of all the external nodes.
We may consider the planar embedding is on the sphere
with one fixed point in the outer face designated as $\infty$.
We temporarily connect each external node to $\infty$ by non-intersecting paths.
As we clockwise-traverse from one external node to the next, if they belong to different components $G_i$ and $G_j$, we can connect one non-external node $u$ from $G_i$ to one non-external node $v$ from $G_j$ by a path of length 2: $u, e=\{u,w\},w,e'=\{w,v\},v$, together with one extra node $w'$ and an edge $\{w,w'\}$.
This gadget can be made disjoint from all the temporary paths to $\infty$,
and also disjoint from each other.
All new edges (there are three edges on each such gadget) have weight 1.
In any perfect matching, $w$ is matched to $w'$ and 
therefore this gadget has no effect on the signature.
Then we remove the temporary paths to $\infty$.
The only purpose is to make the matchgate graph
(1) connected, and (2) its outer face uniquely well-defined 
for the given planar embedding, with a connected boundary.

Lastly, we concern ourselves with the orientation $\orientg$ for $G$.
The $k$ external nodes are labeled clockwise 
 $1$ through $k$ exactly in that order.
When we index $G$ with a length-$k$ bitstring $\alpha$, 
 the bits in $\alpha$ refer to the external nodes in $G$ 
in this clockwise order.
The neighbor $\hat{i}$ of $i$ is labeled $2k+1 -i$.
Then we let $\orientg$ be any Pfaffian orientation of $G$.
Note that the orientation of \emph{bridge} edges (edges 
that are not part of any cycle) have no bearing on
the orientation being a Pfaffian orientation, and therefore can be
arbitrary.
In particular, for each $\{i,\hat{i}\}$ edge 
(being a bridge edge) we assume it is oriented in the order $(i, \hat{i})$:
 from low to high.

With our graph so preprocessed, we are ready to prove our theorem.
For every bit values $b$ and $c$, and strings $u, v, w$,
 for brevity we will use $G^{bc}$ to refer to the graph $G^{ubvcw}$,
suppressing $u, v, w$.
\begin{proof}[Proof for the case $b=0,c=0$]
We assume that $\delta(u0v0w)$ and $\delta(u1v1w)$ are both defined
(for the particular $u, v, w$).

Hence there exists a perfect matching in $G^{00}$, call it $M^{00}$.
Similarly there exists a perfect matching in $G^{11}$, call it $M^{11}$.
Let $e^*=\{i,j\}$ be a new edge (recall that $i<j$ are the external nodes 
referenced by $b,c$).
This is an undirected edge placed in the outer face.
Define the graph $G^*=G^{00}\cup\{e^*\}$, having the same set of vertices as $G^{00}$
and one extra edge $e^*$. See Fig. \ref{fig:estarcase1}.
This introduces a new non-outer face, consisting of the segment from external nodes $i$ to $j$ (corresponding to $b v c$) followed by $e^*$.
The segment has a path from $i$ to $j$  through all the external
nodes $\ell$ or its neighbor  $\hat{\ell}$ referenced in $v$ because
the boundary of the outer face is connected.
Viewed from within the new non-outer face just formed by this path and $e^*$,
the segment $bvc$  is traversed in counter-clockwise direction.

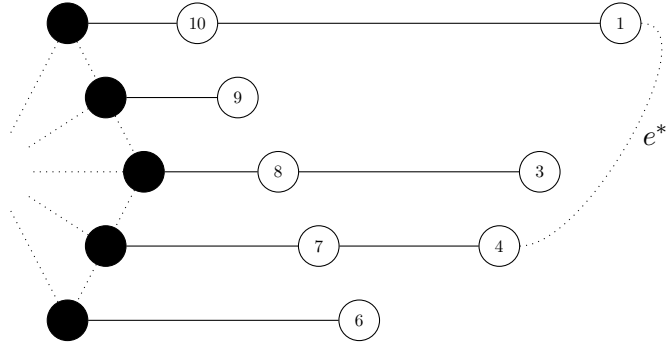
\begin{figure}
\centering
\begin{tikzpicture}
[vertex/.style={circle,draw,scale=1,minimum size=.9cm},input/.style={circle,draw},scale=0.6,transform shape,invis/.style={circle,scale=1,minimum size=.9cm}]
\node[circle](scaffold1){};
\node[circle](scaffold2)[right=2.5cm of scaffold1]{};
\node[circle](scaffold3)[right=9cm of scaffold2]{};

\node[vertex,black,fill](1)[below=0cm of scaffold1]{};
\node[vertex,black,fill](2)[below right=1cm and 0.2cm of 1]{};
\node[vertex,black,fill](3)[below right=1cm and 0.2cm of 2]{};
\node[vertex,black,fill](4)[below left= 1cm and 0.2cm of 3]{};
\node[vertex,black,fill](5)[below left= 1cm and 0.2cm of 4]{};
\node[circle,scale=5](core)[left=2cm of 3]{};

\node[vertex](a1)[below=0cm of scaffold2]{$10$};
\node[vertex](a2)[below right=1cm and 0.25cm of a1]{$9$};
\node[vertex](a3)[below right=1cm and 0.25cm of a2]{$8$};
\node[vertex](a4)[below right=1cm and 0.25cm of a3]{$7$};
\node[vertex](a5)[below right=1cm and 0.25cm of a4]{$6$};

\node[vertex](b1)[below=0cm of scaffold3]{$1$};
\node[invis](b2)[below left=1cm and 0.25cm of b1]{};
\node[vertex](b3)[below left=1cm and 0.25cm of b2]{$3$};
\node[vertex](b4)[below left=1cm and 0.25cm of b3]{$4$};
\node[invis](b5)[below left=1cm and 0.25cm of b4]{};

\foreach \p in {1,2,3,4,5} {
\draw [-] (\p) to (a\p);
}
\foreach \p in {1,3,4} {
\draw [-] (a\p) to (b\p);
}
\foreach \p/\q in {1/2,2/3,3/4,4/5} {
\draw [dotted] (\p) to (\q);
}
\foreach \p in {1,2,3,4,5} {
\draw [dotted] (\p) to (core);
}
\draw (b1) edge[dotted,out=0, in=0, looseness=0.7] node[scale=1.66,right] {$e^*$} (b4);
\end{tikzpicture}
\caption{Adding $e^*$ in the $b=0,c=0$ case. In this example $i=1$, $j=4$, and $bvc=0100$.}\label{fig:estarcase1}
\end{figure}

By adding $e^*$ we have exactly one more face in $G^*$ compared to $G$,
as well as compared to  $G^{00}$ and $G^{11}$.
Let $\orientgalpha[*]=\orientgalpha[00] \cup \{\overrightarrow{e^*}\}$
 with $\overrightarrow{e^*}$ oriented
either as $(i,j)$ or as $(j,i)$,
  such that that new face has
an odd number of  clockwise oriented edges, 
as demanded by Kasteleyn's algorithm to produce
a Pfaffian orientation.
We note that each existing bridge edge $\{\ell, \hat{\ell}\}$ 
corresponding to a bit 0 in $v$
contributes exactly one extra clockwise oriented edge in the traversal
around the boundary of the new face, since  it is traversed in both
directions exactly once.
We define $M^*=M^{11}\cup \{e^*\}$.
Note that $M^*\in\mathcal M(G^*)$ and
we may also  consider $M^{00}\in\mathcal M(G^*)$.

We shall use $M^*$ as an intermediate step to understand 
how $\delta(u0v0w)$ and $\delta(u1v1w)$ are related.
Our goal is to show 
that 
their product is a function entirely of $i,j$ and $\orientg$, 
and independent of $u$, $v$ and $w$,
thus proving our theorem.

\noindent
{\bf Claim:}
The signs of $M^*$ and $M^{00}$ are the same.
Formally:
\begin{equation}
\frac{\Pf_{\orientgalpha[00]}(M^{00})}{\Gamma_{G^{00}}(M^{00})}
=
\frac{\Pf_{\orientgalpha[*]}(M^{00})}{\Gamma_{G^{*}}(M^{00})}
=
\frac{\Pf_{\orientgalpha[*]}(M^{*})}{\Gamma_{G^{*}}(M^{*})}
\end{equation}
The first equality follows from the fact that adding the edge $e^*$ to $G^{00}$ 
does not change the Pfaffian term nor the perfect matching term,
both corresponding to the perfect matching
$M^{00}$, since $M^{00}$ does not contain the edge $e^*$.
The second equality follows from
 the fact that \emph{all} perfect matchings in a 
Pfaffian-oriented graph must have the same sign.
terms in a Pfaffian-oriented graph must have the same sign.

Now we compare $M^*$ and $M^{11}$.
The perfect matching terms are the same, $\Gamma_{G^{*}}(M^{*})=\Gamma_{G^{11}}(M^{11})$,
since the additional edge $e^*$ has weight 1.
We write out their Pfaffian terms explicitly.
For $\Pf_{\orientgalpha[11]}(M^{11})$ we will write it in the canonical
form where the listing of matched edges are given according to the stipulation
after (\ref{Pfaffian-defined}).  Note that the nodes of $G^{11}$ are
linearly ordered by the induced order from that of $G$.
For $\Pf_{\orientgalpha[*]}(M^{*})$ we will write it
by appending the extra matched pair $\{i, j\}$ in the order $i<j$ at the end.
\begin{align}
\Pf_{\orientgalpha[11]}(M^{11}) &=\epsilon_{\pi_1}A_{x_1,x_2}A_{x_3,x_4}\ldots A_{x_{n-1},x_n}\\
\Pf_{\orientgalpha[*]}(M^{*})   &=\epsilon_{\pi_2}A_{x_1,x_2}A_{x_3,x_4}\ldots A_{x_{n-1},x_n}A_{i,j}
\end{align}
where 
$x_1 < x_2, x_3 < x_4, \ldots, x_{n-1} < x_n$,
$x_1 < x_3 < \ldots < x_{n-1}$. 
$A_{i,j}=\pm1$, and it is $+1$ if 
$\overrightarrow{e^*}$ is oriented as $(i,j)$  and it is $-1$ if
it is oriented as $(j,i)$, according to Kasteleyn's algorithm.
The sign of the permutation
$\epsilon_{\pi_1}$  counts the parity of the overlapping pairs
among $\{\{x_1,x_2\}, \{x_3,x_4\}, \ldots, \{x_{n-1},x_n\} \}$. 
The sign $\epsilon_{\pi_2}$
 counts the parity of the overlapping pairs
among $\{ \{x_1,x_2\}, \{x_3,x_4\}, \ldots, \{x_{n-1},x_n\}, \{i, j\} \}$.
Thus $\epsilon_{\pi_2}/\epsilon_{\pi_1} = (-1)^z$, 
where $z$ is the number of overlaps
between $\{i, j\}$ and the edges in $M^{11}$.

We account for these two sources of change in values separately.

Consider $\epsilon_{\pi_2}/\epsilon_{\pi_1}$.
To form an overlapping pair
  with $\{i, j\}$,  a pair must be an edge with one label
between $i$ and $j$ and one label outside. Vertices with a label
between $i$ and $j$ correspond exactly  to 
the external nodes within the segment $v$ that are not removed.
These external nodes must be matched within $M^{11}$ to a node of
a label greater than $j$.
It follows that $z$ is precisely the  number of 0's within $v$.

Now consider $A_{i,j}$.
It is $-1$ if $e^*$ is oriented high-to-low in $\orientgalpha[*]$, and $1$ otherwise.
Let $f(\orientg,i,j)$ be this value
when the orientation of $e^*$ is made to the graph $G^{11\ldots1\oplus e_i\oplus e_j}$---the
graph obtained from $G$ with all
external nodes removed except
 $i$ and $j$---according to Kasteleyn's algorithm.  
Relative to this,
if the orientation of $e^*$ is made to the graph $G^*$,
the orientation is changed according to the parity of the number of zeros in $v$,
the removal pattern within the segment between $i$ and $j$.
More precisely, each $0$ within $v$ adds one more bridge edge
of the form $\{\ell, \hat{\ell} \}$ where $i<\ell<j$ and changes the orientation of $e^*$
exactly once. Hence the value $A_{i,j}$
is precisely $f(\orientg,i,j)(-1)^{z}$,
where, again $z$ is the  number of 0's within $v$.


Returning to our Pfaffian terms:
\begin{align}
\Pf_{\orientgalpha[11]}(M^{11}) &=\epsilon_{\pi_1}A_{x_1,x_2}A_{x_3,x_4}\ldots A_{x_{n-1},x_n}\\
\Pf_{\orientgalpha[*]}(M^{*})   &=\epsilon_{\pi_2}A_{x_1,x_2}A_{x_3,x_4}\ldots A_{x_{n-1},x_n}A_{i,j}\\
                                &=\epsilon_{\pi_1}A_{x_1,x_2}A_{x_3,x_4}\ldots A_{x_{n-1},x_n}f(\orientg,i,j).
\end{align}
Note that the two factors $(-1)^z$ are canceled.
So the sign 
difference between $\delta(u0v0w)$ and $\delta(u1v1w)$ is entirely a function of $\orientg$ and $i,j$, and not the constituent $u,v,w$.
\end{proof}

\begin{proof}[Proof for the case $b=1,c=0$]
Following the idea of $M^{00}$ and $M^{11}$ from the previous proof, we define $M^{01}$ and $M^{10}$ in the graphs $G^{01}, G^{10}$ respectively.
Recall that the neighbor of $i$ is $\hat{i}$, and the neighbor of $j$ is $\hat{j}$, $\{i,\hat{i}\} \in M^{01}$ and $\{j, \hat{j}\} \in M^{10}$,
and that they are specially labeled such that $i<j<\hat{j}<\hat{i}$.
We define a new edge  $e^*=\{j,\hat{i}\}$, and
$G^* = G^{10}\cup \{e^*\}$.
See Fig.\ \ref{fig:estarcase2}.
The edge $e^*$ is drawn on the outer face of $G^{10}$ so that
$G^*$ is a plane graph with one more non-outer face.
We orient $G^*$ to $\orientgalpha[*]$,
namely to orient the edge $e^*$ appropriately 
as before by Kasteleyn's algorithm.
Let $M^*=(M^{01} - \{i,\hat{i}\}) \cup \{e^*\}$.
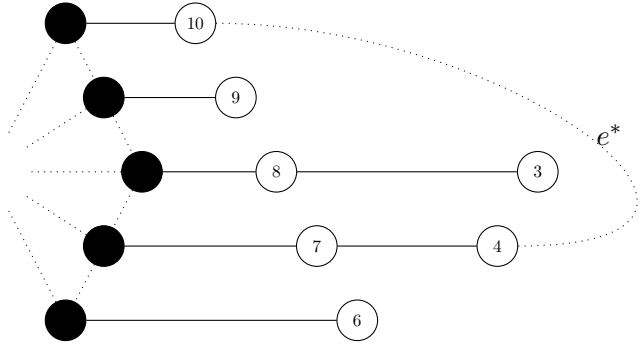
\begin{figure}
\centering
\begin{tikzpicture}
[vertex/.style={circle,draw,scale=1,minimum size=.9cm},input/.style={circle,draw},scale=0.6,transform shape,invis/.style={circle,scale=1,minimum size=.9cm}]
\node[circle](scaffold1){};
\node[circle](scaffold2)[right=2.5cm of scaffold1]{};
\node[circle](scaffold3)[right=9cm of scaffold2]{};

\node[vertex,black,fill](1)[below=0cm of scaffold1]{};
\node[vertex,black,fill](2)[below right=1cm and 0.2cm of 1]{};
\node[vertex,black,fill](3)[below right=1cm and 0.2cm of 2]{};
\node[vertex,black,fill](4)[below left= 1cm and 0.2cm of 3]{};
\node[vertex,black,fill](5)[below left= 1cm and 0.2cm of 4]{};
\node[circle,scale=5](core)[left=2cm of 3]{};

\node[vertex](a1)[below=0cm of scaffold2]{$10$};
\node[vertex](a2)[below right=1cm and 0.25cm of a1]{$9$};
\node[vertex](a3)[below right=1cm and 0.25cm of a2]{$8$};
\node[vertex](a4)[below right=1cm and 0.25cm of a3]{$7$};
\node[vertex](a5)[below right=1cm and 0.25cm of a4]{$6$};

\node[invis](b1)[below=0cm of scaffold3]{};
\node[invis](b2)[below left=1cm and 0.25cm of b1]{};
\node[vertex](b3)[below left=1cm and 0.25cm of b2]{$3$};
\node[vertex](b4)[below left=1cm and 0.25cm of b3]{$4$};
\node[invis](b5)[below left=1cm and 0.25cm of b4]{};

\foreach \p in {1,2,3,4,5} {
\draw [-] (\p) to (a\p);
}
\foreach \p in {3,4} {
\draw [-] (a\p) to (b\p);
}
\foreach \p/\q in {1/2,2/3,3/4,4/5} {
\draw [dotted] (\p) to (\q);
}
\foreach \p in {1,2,3,4,5} {
\draw [dotted] (\p) to (core);
}
\draw (a1) edge[dotted,out=0, in=0, looseness=2] node[scale=1.66,right] {$e^*$} (b4);
\end{tikzpicture}
\caption{Adding $e^*$ in the $b=1,c=0$ case. In this example $\hat{i}=10$, $j=4$, and $bvc=1100$.}\label{fig:estarcase2}
\end{figure}

We claim, using the same
reasoning from the previous proof, that $M^*$ and $M^{10}$ have the same sign.

\begin{equation}
\frac{\Pf_{\orientgalpha[10]}(M^{10})}{\Gamma_{G^{10}}(M^{10})}
=
\frac{\Pf_{\orientgalpha[*]}(M^{10})}{\Gamma_{G^{*}}(M^{10})}
=
\frac{\Pf_{\orientgalpha[*]}(M^{*})}{\Gamma_{G^{*}}(M^{*})}
\end{equation}
The first equality is because $M^{10}$ does not contain the edge $e^*$ 
which was added to $G^{10}$ to obtain $G^*$.
The second equality is because $M^{10}$ and $M^{*}$ are both
perfect matchings in a Pfaffian oriented graph $\orientgalpha[*]$.

Now we only need to compare $M^*$ and $M^{01}$.
The perfect matching terms are the same, $\Gamma_{G^{*}}(M^{*})=\Gamma_{G^{01}}(M^{01})$,
since both edges $e^*$ and $\{i,\hat{i}\}$ have weight 1.
We shall use the same approach as before to analyze the Pfaffian terms.
Once again we write their Pfaffian terms explicitly:
\begin{align}
\Pf_{\orientgalpha[01]}(M^{01}) &=\epsilon_{\pi_1}A_{x_1,x_2}A_{x_3,x_4}\ldots A_{x_{n-1},x_n}A_{i,\hat{i}}\\
\Pf_{\orientgalpha[*]}(M^{*})   &=\epsilon_{\pi_2}A_{x_1,x_2}A_{x_3,x_4}\ldots A_{x_{n-1},x_n}A_{j,\hat{i}}
\end{align}
where the labels of the matching edges satisfy
$x_1 < x_2, x_3 < x_4, \ldots, x_{n-1} < x_n$,
$x_1 < x_3 < \ldots < x_{n-1}$,
and $\epsilon_{\pi_1}$ and $\epsilon_{\pi_2}$ count the parity of
the number of overlapping pairs
 among the matching edges in $M^{01}$ and  $M^*$ respectively.
To compute $\epsilon_{\pi_2}/\epsilon_{\pi_1}$, we only need to account
for the parities of
the number of overlapping pairs between $\{i,\hat{i}\}$ and 
the other matching edges in $M^{01}$,
and between  $\{j, \hat{i}\}$ and the other matching edges in $M^{*}$.
As a necessary condition, any such an overlapping edge must have at least one end point strictly less than $\hat{i}$.
Let us account for all edges $\{x, y\}$ with the minimum label $\min\{x, y\}
< \hat{i}$. Those with $\min\{x, y\} <i$ are not overlapping edges since
each has a unique neighbor with label greater than $\hat{i}$.
The unique edge with $\min\{x, y\} =i$ is $\{i,\hat{i}\}$ in $G^{01}$, which
is not present in $G^*$.
The edges with $i < \min\{x, y\}<j$ are of the form $\{\ell, \hat{\ell}\}$.
They are in 1-1 correspondence
with the 0's in $v$, 
and do contribute  an overlapping pair  in   $M^*$ but {\it not} in $M^{01}$.
The vertex $j$ is not present in $G^{01}$, and the edge $e^* =
\{j, \hat{i}\}$ has $\min\{x, y\} =j$, and 
is in $M^*$.
All edges with $j < \min\{x, y\} < \hat{i}$ either do not contribute
to an overlap
 in both  $M^{01}$ and  $M^*$ (when $j < \min\{x, y\}  \le k$),
or  do contribute
to an overlap in both  $M^{01}$ and  $M^*$ (when $k < \min\{x, y\} < \hat{i}$).
The conclusion is that
$\epsilon_{\pi_2}/\epsilon_{\pi_1} = (-1)^z$,
where $z$ is the number of 0's in $v$.

%

Now consider $A_{i,\hat{i}}$ and $A_{j,\hat{i}}$.
Because we oriented the bridge edge $\{i,\hat{i}\}$ from low to high, we know
that  $A_{i,\hat{i}}$ is  $1$.
We now need only consider $A_{j,\hat{i}}$.
By the  same reasoning  to the previous proof, we conclude
\begin{equation}
A_{j,\hat{i}}=f(\orientg,i,j)(-1)^{z},
\end{equation}
where $f(\orientg,i,{j})$ is the $\pm 1$ value for $A_{j,\hat{i}}$
when we introduce the edge $\{j, \hat{i}\}$ to the graph
obtained from $G$ with all
external nodes removed {\it except} $j$, according to Kasteleyn's algorithm.



Our conclusion is the same:
\begin{align}
\Pf_{\orientgalpha[01]}(M^{01}) &=\epsilon_{\pi_1}A_{x_1,x_2}A_{x_3,x_4}\ldots A_{x_{n-1},x_n}A_{i,\hat{i}}\\
                                &=\epsilon_{\pi_1}A_{x_1,x_2}A_{x_3,x_4}\ldots A_{x_{n-1},x_n}\\
\Pf_{\orientgalpha[*]}(M^{*})   &=\epsilon_{\pi_2}A_{x_1,x_2}A_{x_3,x_4}\ldots A_{x_{n-1},x_n}A_{j,\hat{i}}\\
                                &=\epsilon_{\pi_1}A_{x_1,x_2}A_{x_3,x_4}\ldots A_{x_{n-1},x_n}f(\orientg,i,j),
\end{align}
again with the two factors $(-1)^z$ canceled.
The second line follows from the fact that $A_{i,\hat{i}} = 1$.
Again we conclude that the difference in sign between $ubvcw$ and $u\overline{b}v\overline{c}w$ is entirely a function of $\orientg$ and
 $i,j$, and not of $u,v,w$.
\end{proof}
With this, the proof of Theorem~\ref{thm:mgi} is complete, namely
(planar) matchgate signatures satisfy the Matchgate Identities.



\section{MGI Imply Matchgate-Realizable}\label{sec:realizing}
Any signature of a 
matchgate must satisfy the Matchgate Identities.
In this section,
 we show that any $\Gamma\in(\mathbb C^2)^{\tensor k}
= \mathbb{C}^{2^k}$ satisfying the Matchgate Identities can be realized
as the signature of a 
matchgate with $k$ external nodes.
Thus MGI are not only necessary but also sufficient conditions
for matchgate signatures.

Consider a length $2^k$ vector $\Gamma$ indexed by $\{0,1\}^k$ satisfying MGI.
If it is the all-zeros vector then it is trivially realizable.
So assume there is at least one non-zero value.

\paragraph{Preprocessing}
Assume $\Gamma^\beta\neq0$, for some $\beta \in \{0,1\}^k$.
Define $\Gamma'^{\alpha}=\Gamma^{\alpha\oplus\overline{\beta}}/\Gamma^\beta$,
where $\overline{\beta} = \beta \oplus 11\ldots1$.
Thus, $\Gamma'^{11\ldots1}=1$, and $\Gamma'$ also satisfies MGI.
In this section we will create a matchgate $G'$ with signature $\Gamma'$.
Given such a $G'$, we can create a matchgate $G$
with signature $\Gamma$ as follows: 
First we add two new non-external nodes $u,v$ to $G'$ and an edge $\{u,v\}$ of weight $\Gamma^\beta$.
Those two nodes are not connected to any other nodes---in effect they contribute exactly a factor $\Gamma^\beta$ to each
perfect matching term.
Then, if the $i$-th bit of $\overline{\beta}$ is one, we 
add a new edge $\{v_i, v'_i\}$ of weight one to the $i$-th external node $v_i$,
and making $v'_i$ the new $i$-th external node.
It follows that the signature of $G$ is exactly $\Gamma$.

\paragraph{Construction}
We now show that we can realize $\Gamma$ satisfying MGI
and  $\Gamma^{11\ldots1}=1$.
Let $K_k$ denote the complete graph on $k$ vertices.
The labels of $K_k$ are ordered $1 < 2 < \ldots < k$,
and correspond to the bit positions in the
index for $\Gamma$.
We place the nodes of $K_k$ on a convex curve,
as illustrated in Fig.\ \ref{fig:k5}.
The nodes are arranged in clockwise order by their index, 
and two edges cross each other geometrically in the drawing of the graph
iff their labels form an overlapping pair  as defined before algebraically.
(We assume the $k$ nodes are placed in general position, so that
any pair of crossing edges intersect at a unique point. There are exactly
$\binom{k}{4}$ such intersection points.)
For each $\alpha$ of Hamming weight $k-2$, note that $K_k^\alpha$ has exactly one edge left.
For each such $\alpha$, set the weight of the unique edge in $K_k^\alpha$ 
to be $\Gamma^\alpha$.
This defines a weight for every edge of $K_k$.

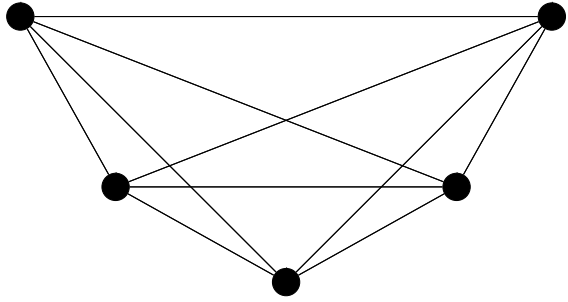
\begin{figure}
\centering
\begin{tikzpicture}
[vertex/.style={circle,draw,fill=black},input/.style={circle,draw}]
\node[vertex](1){};
\node[vertex](2)[below right=2cm and 1cm of 1]{};
\node[vertex](3)[below right=1cm and 2cm of 2]{};
\node[vertex](4)[above right=1cm and 2cm of 3]{};
\node[vertex](5)[above right=2cm and 1cm of 4]{};

\foreach \p in {1,2,3,4,5} {
\foreach \q in {1,2,3,4,5} {
\draw [-] (\p) to (\q);
}
}
\end{tikzpicture}
\caption{The embedding for $K_5$.}\label{fig:k5}
\end{figure}

\paragraph{Equality with Pfaffian}
Clearly our embedding of $K_k$ is not planar for a general $k \ge 4$.
We first prove the following equality:
Let $\Pf(K_k^{\alpha})$ be the Pfaffian value of the skew-symmetric matrix
representing $K_k^{\alpha}$ where the nodes of $K_k^{\alpha}$
have the induced order from $1 < 2 < \ldots < k$.
Then 
for all $\alpha \in \{0, 1\}^k$:
\begin{equation}\label{k-choose-2-determine-Gamma}
\Pf(K_k^{\alpha}) = \Gamma^\alpha.
\end{equation}
It follows that the $\binom{k}{2}$ edge weights of $K_k$ determine
the $2^k$ values of any $\Gamma$ satisfying MGI.

Clearly (\ref{k-choose-2-determine-Gamma})
holds for any $\alpha$ of Hamming weight greater than or equal to $k-2$.
By assumption $\Gamma$ satisfies the Matchgate Identities (\ref{eq:mgi}).
Inductively, consider $\alpha$ with Hamming weight $k-l$ for some even $l>2$.
Let $\{p_1, \ldots, p_l\}$  be the set of indices 
listed in increasing order $p_1 < \ldots < p_l$, where $\alpha$ has the bit 0.
 These are the bit positions where  $\alpha$ differs from $1^k$.
Consider the MGI on $\alpha\oplus e_{p_1}$ and $1^k\oplus e_{p_1}$:
\begin{equation}\label{eq:mgiapplied}
\Gamma^\alpha\Gamma^{11\ldots1}=\sum_{i=2}^l(-1)^i\Gamma^{\alpha\oplus e_{p_1}\oplus e_{p_i}}\Gamma^{1^k\oplus e_{p_1} \oplus e_{p_i}}
\end{equation}
As $\Gamma^{11\ldots1}=1$, we see that $\Gamma^{\alpha}$ is defined by higher Hamming weight terms.

Thus all lower Hamming weight terms
of $\Gamma$ are determined by those of weight $k-2$.
However we observe that, by (\ref{eq:gpi2}),
the  Pfaffian values also satisfy
exactly the same identities   as  MGI.
By induction, it follows  that $\Pf(K_k^{\alpha})=\Gamma^\alpha$ 
for all $\alpha$.

\paragraph{Planarizing $K_k$}
We want to show next
that there exists a {\it planar} matchgate $G$ with signature $\Gamma_G=\Gamma$.
We construct such a $G$ from $K_k$.
Consider the convex embedding of $K_k$. For $k \ge 4$
it has some edge crossings, as shown in Fig.\ \ref{fig:k5}.
The planar graph $G$ is created by replacing each edge crossing with a \emph{crossover gadget} from Fig.\ \ref{fig:crossovergadget}.
The crossover gadget is itself a matchgate $X$ with the following signature:
\begin{align}
X^{0000}&=1\\
X^{0101}&=1\\
X^{1010}&=1\\
X^{1111}&=-1
\end{align}
and for all other $\beta \in \{0, 1\}^4$, $X^\beta=0$.
We note that even though geometrically this gadget is only symmetric 
under a rotation of $\pi$ (but not $\pi/2$), its signature
is invariant under a cyclic permutation, and thus functionally
it is symmetric under a rotation of $\pi/2$.
Now we replace every crossing of a pair of edges in the embedded $K_k$
by a copy of $X$.
For example, this replacement by the crossover gadget
changes Fig.\ \ref{fig:k5} to Fig.\ \ref{fig:planarizedk5}.
If an edge $\{i, j\}$ in $K_k$ crosses some other edges (this happens for every
non-adjacent $i$ and $j$ in the cyclic sense), 
then this replacement breaks the edge
$\{i, j\}$ into several parts. If $\{i, j\}$ crosses $t \ge 0$ other edges,
then it is replaced by $t+1$ edges (outside of crossover gadgets)---we
will call them the \emph{$i$-$j$-passage}---in
 addition to $t$ copies
of the crossover gadget. Of course one
 copy of the crossover gadget is used for both
edges of a pair of crossing edges in this replacement.
Define $I$ to be the set of all edges in $G$ 
that are \emph{not} part of a crossover gadget.
Then each edge $\{i, j\}$ in $K_k$ defines a unique subset of
edges in $I$, which is the \emph{$i$-$j$-passage}.
 It is clear that $I$ is
a disjoint union of these $i$-$j$-passages,
over all $\binom{k}{2}$ pairs $1 \le i<j \le k$.
Finally we choose one edge in each $i$-$j$-passage to have
the weight $\Gamma^{[k]-\{i, j\}}$,
namely the edge weight of $\{i, j\}$  in $K_k$.
To be specific,  we will choose this edge to be the
one adjacent to $i$, the lower indexed external node of $\{i, j\}$.
All other edges of $I$ are assigned weight one. See Fig.\ \ref{fig:labeledplanarized}.
This  defines our planar matchgate $G$ with external nodes
 $1 < 2 < \ldots < k$.

\begin{figure}
\centering
\begin{tikzpicture}
[vertex/.style={circle,draw,fill=black},input/.style={circle,draw}]
\node[input](input1){1};


\node[vertex](gadge1)[below=of input1]{};
\node[vertex](gadge2)[right=of gadge1]{};


\node[input](input2)[above=of gadge2]{2};
\node[input](input3)[below=of gadge1]{4};
\node[input](input4)[below=of gadge2]{3};

\path [-] (gadge1) edge node[above] {$-1$} (gadge2);
\draw [-] (input1) -- (input2) -- (gadge2) -- (input4) -- (input3) -- (gadge1) -- (input1);
\end{tikzpicture}
\caption{The crossover gadget. The external nodes are those labeled, and all edge weights are $1$, except the edge labeled $-1$.}\label{fig:crossovergadget}
\end{figure}
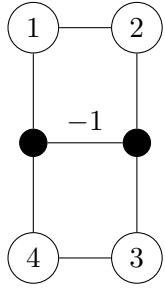

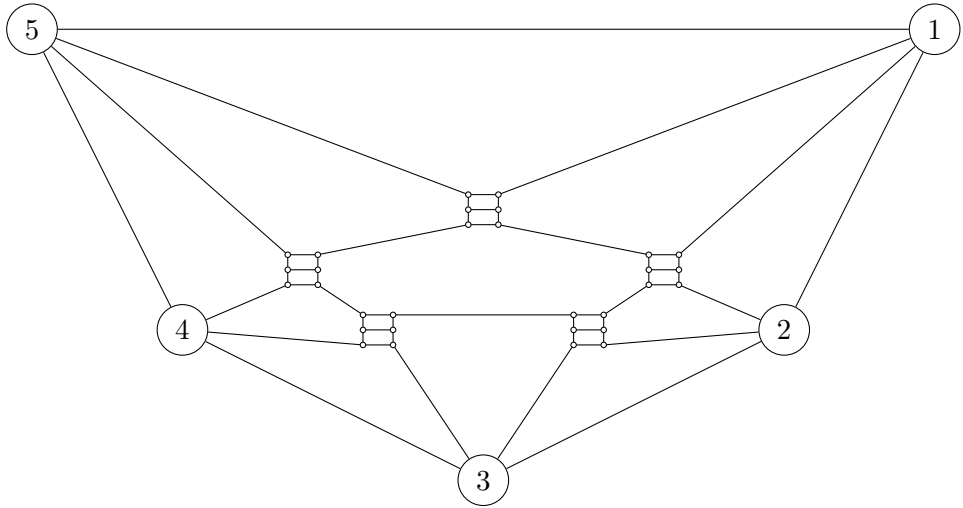
\begin{figure}
\centering
\begin{tikzpicture}
[vertex/.style={circle,draw},input/.style={circle,draw},inner/.style={circle,draw,scale=0.2}]
\node[vertex](1) at (0,6) {5};
\node[vertex](2) at (2,2) {4};
\node[vertex](3) at (6,0) {3};
\node[vertex](4) at (10,2) {2};
\node[vertex](5) at (12,6) {1};

\foreach \p/\q in {1/2,2/3,3/4,4/5,5/1} {
\path [-] (\p) edge (\q);
}

\node[inner](a1) at (5.8 , 3.8) {};
\node[inner](a2) at (6.2 , 3.8) {};
\node[inner](a3) at (5.8 , 3.6) {};
\node[inner](a4) at (6.2 , 3.6) {};
\node[inner](a5) at (5.8 , 3.4) {};
\node[inner](a6) at (6.2 , 3.4) {};

\path [-] (1) edge (a1);
\draw [-] (5) to (a2);

\node[inner](b1) at (3.4,3.0) {};
\node[inner](b2) at (3.8,3.0) {};
\node[inner](b3) at (3.4,2.8) {};
\node[inner](b4) at (3.8,2.8) {};
\node[inner](b5) at (3.4,2.6) {};
\node[inner](b6) at (3.8,2.6) {};

\path [-] (1) edge  (b1);
\path [-] (2) edge  (b5);
\draw [-] (b2) to (a5);

\node[inner](c1) at (4.4,2.2) {};
\node[inner](c2) at (4.8,2.2) {};
\node[inner](c3) at (4.4,2.0) {};
\node[inner](c4) at (4.8,2.0) {};
\node[inner](c5) at (4.4,1.8) {};
\node[inner](c6) at (4.8,1.8) {};

\draw[-] (c1) to (b6);
\path[-] (c5) edge  (2);
\path[-] (c6) edge (3);

\node[inner](d1) at (7.2,2.2) {};
\node[inner](d2) at (7.6,2.2) {};
\node[inner](d3) at (7.2,2.0) {};
\node[inner](d4) at (7.6,2.0) {};
\node[inner](d5) at (7.2,1.8) {};
\node[inner](d6) at (7.6,1.8) {};

\draw[-] (c2) to (d1);
\path[-] (d5) edge  (3);
\draw[-] (d6) to (4);

\node[inner](e1) at (8.2,3.0) {};
\node[inner](e2) at (8.6,3.0) {};
\node[inner](e3) at (8.2,2.8) {};
\node[inner](e4) at (8.6,2.8) {};
\node[inner](e5) at (8.2,2.6) {};
\node[inner](e6) at (8.6,2.6) {};

\draw [-] (5) to (e2);
\draw [-] (4) to (e6);
\draw [-] (e5) to (d2);
\draw [-] (e1) to (a6);

\foreach \a in {a,b,c,d,e} {
	\path[-] (\a3) edge (\a4); 
	\foreach \p/\q in {1/2,5/6} {
	\draw[-] (\a\p) to (\a\q);
	}
	\foreach \p/\q in {1/3,2/4,3/5,4/6} {
	\draw[-] (\a\p) to (\a\q);
	}
}

\end{tikzpicture}
\caption{The graph from Fig.\ \ref{fig:k5} with the crossovers replaced by crossover gadgets from Fig.\ \ref{fig:crossovergadget}.}\label{fig:planarizedk5}
\end{figure}

\begin{figure}
\centering
\begin{tikzpicture}
[vertex/.style={circle,draw},input/.style={circle,draw},inner/.style={circle,draw,scale=0.2}]
\node[vertex](5) at (0,6) {5};
\node[vertex](4) at (2,2) {4};
\node[vertex](3) at (6,0) {3};
\node[vertex](2) at (10,2) {2};
\node[vertex](1) at (12,6) {1};

\foreach \p/\q in {1/5,1/2,2/3,3/4,4/5} {
\path [-] (\p) edge node[above,scale=0.8] {$w(\{\p,\q\})$} (\q);
}

\node[inner](a2) at (5.8 , 3.8) {};
\node[inner](a1) at (6.2 , 3.8) {};
\node[inner](a4) at (5.8 , 3.6) {};
\node[inner](a3) at (6.2 , 3.6) {};
\node[inner](a6) at (5.8 , 3.4) {};
\node[inner](a5) at (6.2 , 3.4) {};

\node[inner](b2) at (8.2,3.0) {};
\node[inner](b1) at (8.6,3.0) {};
\node[inner](b4) at (8.2,2.8) {};
\node[inner](b3) at (8.6,2.8) {};
\node[inner](b6) at (8.2,2.6) {};
\node[inner](b5) at (8.6,2.6) {};

\node[inner](e2) at (3.4,3.0) {};
\node[inner](e1) at (3.8,3.0) {};
\node[inner](e4) at (3.4,2.8) {};
\node[inner](e3) at (3.8,2.8) {};
\node[inner](e6) at (3.4,2.6) {};
\node[inner](e5) at (3.8,2.6) {};

\node[inner](d2) at (4.4,2.2) {};
\node[inner](d1) at (4.8,2.2) {};
\node[inner](d4) at (4.4,2.0) {};
\node[inner](d3) at (4.8,2.0) {};
\node[inner](d6) at (4.4,1.8) {};
\node[inner](d5) at (4.8,1.8) {};

\node[inner](c2) at (7.2,2.2) {};
\node[inner](c1) at (7.6,2.2) {};
\node[inner](c4) at (7.2,2.0) {};
\node[inner](c3) at (7.6,2.0) {};
\node[inner](c6) at (7.2,1.8) {};
\node[inner](c5) at (7.6,1.8) {};

\path [-] (1) edge node[above,scale=0.8] {$w(\{1,4\})$} (a1);
\draw [-] (5) to (a2);

\path [-] (1) edge node[above,scale=0.8] {$w(\{1,3\})$} (b1);
\path [-] (2) edge node[above,scale=0.8] {$w(\{2,5\})$} (b5);
\draw [-] (b2) to (a5);

\draw[-] (c1) to (b6);
\path[-] (c5) edge node[above,scale=0.8] {$w(\{2,4\})$} (2);
\path[-] (c6) edge (3);

\draw[-] (c2) to (d1);
\path[-] (d5) edge node[above,scale=0.8] {$w(\{3,5\})$} (3);
\draw[-] (d6) to (4);

\draw [-] (5) to (e2);
\draw [-] (4) to (e6);
\draw [-] (e5) to (d2);
\draw [-] (e1) to (a6);

\foreach \a in {a,b,c,d,e} {
	\path[-] (\a3) edge (\a4);
	\foreach \p/\q in {1/2,5/6} {
	\draw[-] (\a\p) to (\a\q);
	}
	\foreach \p/\q in {1/3,2/4,3/5,4/6} {
	\draw[-] (\a\p) to (\a\q);
	}
}

\end{tikzpicture}
\caption{The ``planarized'' $K_5$ with edge weights.
The unlabeled edges have weight $1$.}\label{fig:labeledplanarized}
\end{figure}
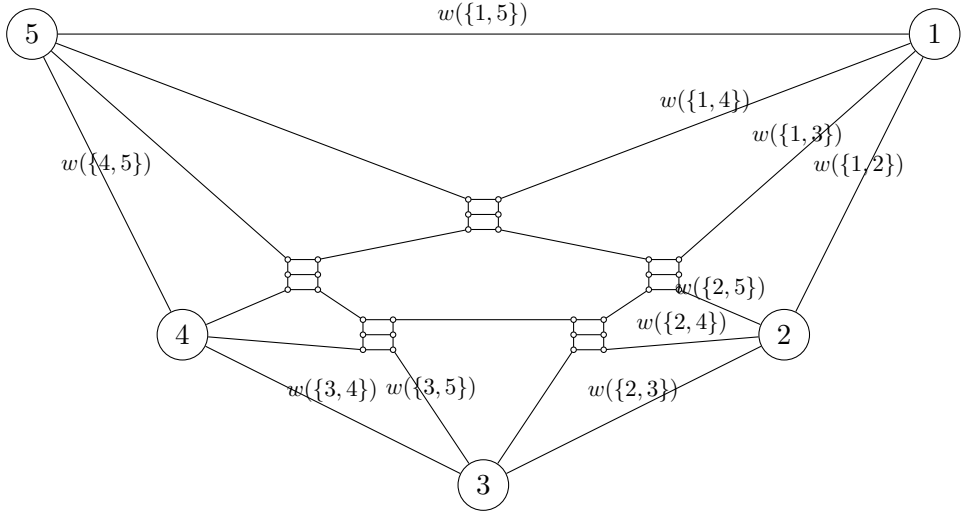

We claim that $\Gamma_G=\Gamma$.

Fix any $\alpha \in \{0, 1\}^k$.
For any $S \subseteq I$,
define ${\mathcal M}_{S}(G^\alpha)$ to be the subset of all 
perfect matchings $M' \in {\mathcal M}(G^\alpha)$ 
such that $M' \cap I= S$.
Every perfect matching $M \in {\mathcal M}(K_k^{\alpha})$
defines a collection  of $i$-$j$-passages, for all $\{i, j\} \in M$.
Let $S(M)$ be the union of these $i$-$j$-passages.
Clearly the 
perfect matching $M \in {\mathcal M}(K_k^{\alpha})$
can be recovered from $S(M)$, and is unique for the given $S(M)$.
There is a 1-1 correspondence between $M$ and $S(M)$.
As an example, we consider $M=\{\{1,3\},\{2,5\}\}\in \mathcal M(K_5^{00010})$.
The set $S(M)$ for $G^{00010}$ is indicated in Fig.\ \ref{fig:highlightedk5}.
\begin{figure}
\centering
\begin{tikzpicture}
[vertex/.style={circle,draw},input/.style={circle,draw},inner/.style={circle,draw,scale=0.2}]
\node[vertex](5) at (0,6) {5};
\node[vertex](3) at (6,0) {3};
\node[vertex](2) at (10,2) {2};
\node[vertex](1) at (12,6) {1};

\foreach \p/\q in {1/2,2/3,5/1} {
\draw [-] (\p) to (\q);
}
\node[inner](a2) at (5.8 , 3.8) {};
\node[inner](a1) at (6.2 , 3.8) {};
\node[inner](a4) at (5.8 , 3.6) {};
\node[inner](a3) at (6.2 , 3.6) {};
\node[inner](a6) at (5.8 , 3.4) {};
\node[inner](a5) at (6.2 , 3.4) {};

\node[inner](b2) at (8.2,3.0) {};
\node[inner](b1) at (8.6,3.0) {};
\node[inner](b4) at (8.2,2.8) {};
\node[inner](b3) at (8.6,2.8) {};
\node[inner](b6) at (8.2,2.6) {};
\node[inner](b5) at (8.6,2.6) {};

\node[inner](e2) at (3.4,3.0) {};
\node[inner](e1) at (3.8,3.0) {};
\node[inner](e4) at (3.4,2.8) {};
\node[inner](e3) at (3.8,2.8) {};
\node[inner](e6) at (3.4,2.6) {};
\node[inner](e5) at (3.8,2.6) {};

\node[inner](d2) at (4.4,2.2) {};
\node[inner](d1) at (4.8,2.2) {};
\node[inner](d4) at (4.4,2.0) {};
\node[inner](d3) at (4.8,2.0) {};
\node[inner](d6) at (4.4,1.8) {};
\node[inner](d5) at (4.8,1.8) {};

\node[inner](c2) at (7.2,2.2) {};
\node[inner](c1) at (7.6,2.2) {};
\node[inner](c4) at (7.2,2.0) {};
\node[inner](c3) at (7.6,2.0) {};
\node[inner](c6) at (7.2,1.8) {};
\node[inner](c5) at (7.6,1.8) {};

\draw [-] (1) to (a1);
\draw [thick] (5) to (a2);

\draw [thick] (1) to (b1);
\draw [thick] (2) to (b5);
\draw [thick] (b2) to (a5);

\draw[thick] (c1) to (b6);
\draw[-] (c5) to (2);
\draw[thick] (c6) to (3);

\draw[-] (c2) to (d1);
\draw[-] (d5) to (3);

\draw [-] (5) to (e2);
\draw [-] (e5) to (d2);
\draw [-] (e1) to (a6);

\foreach \a in {a,b,c,d,e} {
	\path[-] (\a3) edge (\a4); 
	\foreach \p/\q in {1/2,5/6} {
	\draw[-] (\a\p) to (\a\q);
	}
	\foreach \p/\q in {1/3,2/4,3/5,4/6} {
	\draw[-] (\a\p) to (\a\q);
	}
}

\end{tikzpicture}
\caption{The thick edges comprise $S(M)$ for $G^{00010}$,
where $M=\{\{1,3\},\{2,5\}\}$.}\label{fig:highlightedk5}
\end{figure}
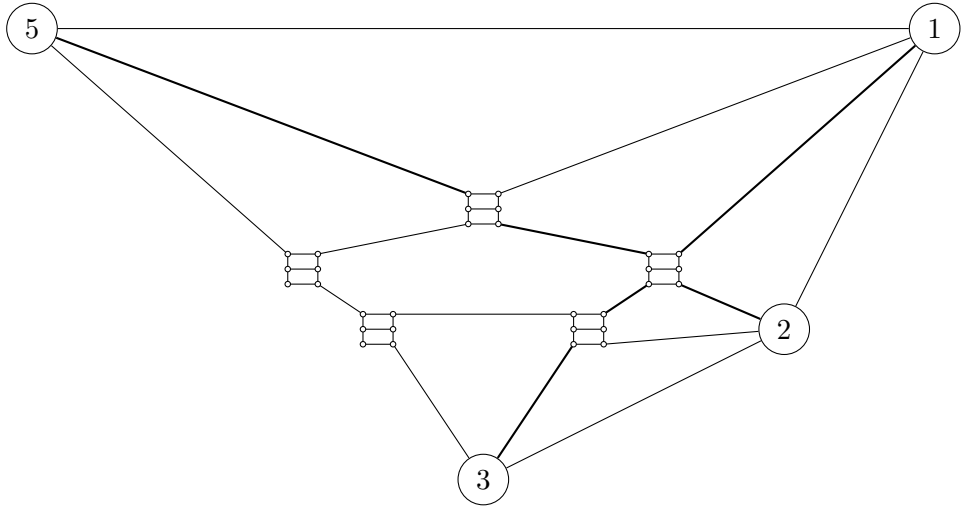

We will show that, for the purpose of computing 
the signature entry $\Gamma_G^\alpha$, we only   need to consider
those perfect matchings $M' \in {\mathcal M}(G^{\alpha})$
that satisfy the following property:
\begin{quote}
{\bf Property:} There exists an $M \in {\mathcal M}(K_k^{\alpha})$,
such that 
\begin{equation}\label{proper-Ms}
M' \cap I = S(M).
\end{equation}
\end{quote}

This is a consequence of properties of the crossover gadget.
If $i$ is an external node in $G^{\alpha}$,
then any $M' \in {\mathcal M}(G^{\alpha})$ must contain
a unique edge $e'$ adjacent to $i$.
There is a unique $j$, which is another external node in $G$,
such that $e'$ belongs  to the $i$-$j$-passage.
Then by the properties of the crossover gadgets
along this $i$-$j$-passage, we may assume $M'$ contains all
edges of this $i$-$j$-passage, saturating $j$.
In particular $j$ belongs to $G^{\alpha}$. All other $M'$
collectively contribute 0, since the evaluation of the
crossover gadget $X$ will be 0.
More generally, in the computation of $\Gamma_G^{\alpha}
= \sum_{M' \in {\mathcal M}(G^{\alpha})} \prod_{e'\in M'}w(e')$,
we classify all $M' \in {\mathcal M}(G^{\alpha})$ according to
$M' \cap I$.  If $S \not = S(M)$ for any $M \in 
{\mathcal M}(K_k^{\alpha})$, then
\begin{equation}\label{improperM-sumstozero}
\sum_{M' \in {\mathcal M}(G^{\alpha}): ~M' \cap I = S~~} 
\prod_{e'\in M'}w(e') = 0.
\end{equation}
In fact, for any $M' \in {\mathcal M}(G^{\alpha})$ such that
$M' \cap I = S$ which is not $S(M)$ 
for any $M \in  {\mathcal M}(K_k^{\alpha})$, 
it must be the case that at some crossover gadget $X$,
$S$ induces an external removal pattern $\beta \not \in
\{0000, 0101, 1010, 1111\}$. Then $X^\beta =0$, and
(\ref{improperM-sumstozero}) follows.

Thus we restrict to those 
perfect matchings $M' \in {\mathcal M}(G^{\alpha})$
that satisfy the property (\ref{proper-Ms}).
For any $M \in {\mathcal M}(K_k^\alpha)$, it is clear that
\begin{equation}\label{c(M)-sign}
\sum_{M'\in {\mathcal M}_{S(M)}(G^\alpha)}      \prod_{e'\in M'}w(e')
= (-1)^{c(M)} \prod_{e \in M} w(e),
\end{equation}
where $c(M)$ counts the number of copies of $X$ where
the external removal pattern is $\beta = 1111$.  Thus
 $c(M)$  is exactly the number of overlapping pairs in $M$.
It follows that
\begin{eqnarray}\label{eq:planarizinglemmaeq1}
\Gamma_G^\alpha
&=&\sum_{M'\in {\mathcal M}(G^\alpha)}\prod_{e'\in M'}w(e')\\
&=&\sum_{S \subseteq I} \sum_{M'\in {\mathcal M}_{S}(G^\alpha)}
      \prod_{e'\in M'}w(e')\\
&=& \sum_{M \in {\mathcal M}(K_k^\alpha)}
\sum_{M'\in {\mathcal M}_{S(M)}(G^\alpha)}      \prod_{e'\in M'}w(e')\\
&=& \Pf(K_k^\alpha).
\end{eqnarray}
The last equality is because each Pfaffian term in $\Pf(K_k^\alpha)$
has exactly the same sign as in (\ref{c(M)-sign}).
Hence $\Gamma_G=\Gamma$
 follows from this and (\ref{k-choose-2-determine-Gamma}).

\begin{theorem}\label{MGI-is-sufficient}
Any 
 $\Gamma\in(\mathbb C^2)^{\tensor k}$ satisfying the Matchgate Identities 
is the signature of a
matchgate with $k$ external nodes.
The matchgate has $O(k^4)$ nodes.
If $\Gamma^{11\ldots 1} =1$, achievable by a normalization
for every nonzero $\Gamma$, there exists a skew-symmetric
matrix $M \in \mathbb{C}^{k \times k}$ such that
$\Gamma^\alpha = \Pf(M^\alpha)$,
where $M^\alpha$ is the matrix obtained from $M$ by deleting
all rows and columns belonging to the subset denoted by $\alpha$.
\end{theorem}

\section{Character}\label{Character}
In \cite{valiant02a} Valiant showed that a fragment
 of quantum computation could be simulated in polynomial time through the \emph{character} of general (not-necessarily-planar) matchgates.
The notion of a general matchgate and its character ultimately inspired planar matchgates and their signatures.
The character is directly based on the notion of the Pfaffian, and what counting problems are expressible in that form.
This section will be concerned with characters and general matchgates.
It will conclude by proving that characters of general matchgates are essentially equivalent to signatures of planar matchgates.

\subsection{Definitions}
\paragraph{The Pfaffian of an Undirected Graph}
For an undirected, labeled, weighted graph $G=(V,E,W)$ there is
 a skew-symmetric matrix $M_G$.
For $i<j$, we define $(M_G)_{i,j}=w(\{i,j\})$, the weight of the edge $\{i,j\}\in E$.
If that edge does not exist, we say the weight is 0.
For $i>j$, we define $(M_G)_{i,j}=-w(\{i,j\})$.
We define $\Pf(G)=\Pf(M_G)$.

\paragraph{General Matchgate}
A general matchgate $G=(V,E,W)$ is an undirected,
labeled, weighted  graph with three designated subsets of $V$. 
The set $X\subseteq V$ is the set of \emph{input} nodes, the set $Y\subseteq V$ is the set of \emph{output} nodes, and the set $T\subseteq V$ is the set of omittable nodes.
These three subsets are disjoint.
The nodes in $X \cup Y$ are called  \emph{external nodes}.
They also define a (possibly nonempty) fourth subset $U=V-(X\cup Y\cup T)$.

The ordered labeling of the nodes of $G$ obey some rules:
$\forall i \in X,~ \forall j\in T : i < j$, and $\forall j\in T,~ \forall \ell\in Y : j < \ell$.
In other words, ordered from low-to-high, the input nodes $X$ come first, then the omittable nodes $T$, and finally the output nodes $Y$.
The remaining nodes can be interspersed throughout the ordering.

\paragraph{The Omittable Nodes}
For $G$ with a set of omittable nodes $T$, 
we define the ``Pfaffian Sum'', $\PfS$, as follows:
\begin{equation}
\PfS(G) = \sum_{W\subseteq T}\Pf(G-W)
\end{equation}
where the sum is over all subsets $W$ of $T$, and $G-W$ is the graph 
obtained from $G$ with all nodes in $W$ and their incident edges removed.

We can express this solely in terms of Pfaffians as well.
Let $I$ be the index set of $M_G$.
Define $\lambda_i=1$ if $i$ is an index corresponding to an omittable node,
and $\lambda_i=0$ otherwise.
Then,
\begin{equation}
\PfS(G) = \sum_{A\subseteq I}(\prod_{i \in A} \lambda_i)\Pf(M_G[A])
\end{equation}
where the sum is over all subsets $A$ of $I$, and $M_G[A]$ is the matrix
obtained from  $M_G$ with the rows and columns indexed by $A$ removed.
It was shown in \cite{valiant02a} that, for a size-$n$ graph:
\begin{equation}
\PfS(G) = \begin{cases}
\Pf(M_G+\Lambda^{(n)}) &\quad\text{if $n$ even}\\
\Pf(M_G^{+}+\Lambda^{(n+1)}) &\quad\text{if $n$ odd}\end{cases}
\end{equation}
where $\Lambda^{(n)}$ is a simple matrix constructed from the $\lambda_i$ values
and $M_G^{+}$ is $M_G$ with an additional final row, column of all zeros.
Thus the Pfaffian Sum $\PfS(G)$ is also computable in
polynomial time.
This ``omittable node'' feature seems to be quite different from what 
has been presented for planar matchgate signatures.
However, we shall see that it ultimately does not add more power.

\paragraph{The Character of a Matchgate}
Consider  any $Z\subseteq X\cup Y$, a subset of 
the \emph{external nodes} of $G$.
A general matchgate is ultimately part of a larger \emph{matchcircuit}, and the 
external nodes in $G$ are connected to external edges.
The following is from \cite{valiant02a}, ``[w]e consider 
there to exist one external edge from
 each node in $X\cap Z$ and from each node in $Y\cap Z$.
The other endpoint of each of the former is some node of lower index than any in $V$ and of each of the latter is some node of index higher than any in $V$.''

The character of a matchgate $G$ is defined as
\begin{equation}
\chi(G,Z) = \mu(G,Z)\PfS(G-Z).
\end{equation}
The term $\mu(G,Z)$ is called the \emph{modifier value}.
It is one of $\pm 1$, and corresponds to the parity of the overlapping pairs
 between matching edges in $E$ and external edges.
Recall that the number of
overlapping pairs is computed as a function of node labels.
Due to the rules of index ordering, this value is determined by $(G,Z)$,
and is independent of the particular matching in $\PfS(G-Z)$.
Thus $\mu(G,Z)$ is  well-defined for any $(G,Z)$.

We also define the \emph{naked character} $\check{\chi}$ of a matchgate,
without the modifier.
\begin{equation}
\check{\chi}(G,Z) = \PfS(G-Z).
\end{equation}

For brevity and consistency, we 
write $\chi^{\alpha}_G = \chi(G,Z)$, where $\alpha$ is the characteristic bitstring of $Z$.
The naked character will be referred to as $\check{\chi}^{\alpha}_G$.
Where $G$ is clear we may omit it.

\paragraph{Matchcircuits}
These matchgates, not necessarily planar,
 were designed to show that the evaluation of some quantum circuits could be done in polynomial time.
Matchgates can be combined into \emph{matchcircuits}
in specific ways.
The composition is helped by the modifiers; in fact their
sole purpose is to make this composition nicely expressible as
a Pfaffian. 
We will not go into this detail; please see~\cite{valiant02a}.
From another perspective, a matchcircuit is simply a larger matchgate with a modifier value set to a constant $1$, as there are no more edges
external to the entire matchcircuit.
The naked character of a matchcircuit is its character.

\subsection{Equivalence of Naked Characters and Signatures}

We will prove the following theorem:
\begin{theorem}
For a general matchgate $G$ with $k$ external nodes,
there exist two planar matchgates $G_1$
 and $G_2$ such that for all $\alpha \in \{0, 1\}^k$,
\begin{equation}
\check{\chi}_G^\alpha=\Gamma_{G_1}^\alpha+\Gamma_{G_2}^\alpha
\end{equation}
\end{theorem}
\begin{proof}
$G$ may not be a planar graph. We draw it by placing
its nodes on a semi-circle arc.
The nodes appear in a clockwise ordering,
 ordered exactly by their labels in the graph.
The edges are drawing as chords inside the semi-circle arc.
If we place the nodes in general position,
then any pair of intersecting chords intersect at a unique point.
Observe that two edges $(u,v),(x,y)$, where $u<v$ and $x<y$,
  cross in the drawing exactly when $u<x<v<y$ or $x<u<y<v$,
i.e., exactly when they form an overlapping pair. 
This arrangement is very similar to the planar matchgate construction in Section \ref{sec:realizing}.


We start by replacing every crossing of chords by the planar crossover gadget from Fig.\ \ref{fig:crossovergadget}.
For the purpose of this proof, we may consider $G$ as 
a subgraph of some $K_n$. After each crossing has been replaced by the
crossover gadget we have a planar matchgate $G'$. We consider 
 $X \cup T \cup Y$ as its external nodes. Let $\Gamma'$ be  its signature.
Let $\alpha \in \{0, 1\}^{|X \cup T \cup Y|}$ indicate a bit removal pattern, and let
$\beta$ and $\gamma$ be its restrictions to $X \cup Y$
and $T$ respectively.
The same proof in Section~\ref{sec:realizing}
shows that  
\begin{equation}\label{transition-to-signature}
\Gamma'^{\alpha} = \Pf(G^{\beta} - W_{\gamma}),
\end{equation}
where $W_{\gamma}$ is the subset of $T$ indicated by $\gamma$.


Fix any $\beta \in \{0, 1\}^{|X \cup Y|}$, such that
$G^{\beta}$ has an even number of nodes. Then we only need
to consider $\gamma \in \{0, 1\}^{|T|}$ of even Hamming weight in the
sum (\ref{transition-to-signature}). Similarly, if
$G^{\beta}$ has an odd number of nodes, then
 we only need
to consider $\gamma$ of odd Hamming weight
in (\ref{transition-to-signature}).
 
The following idea is from \cite{valiant08} (p. 1952).
There exists a planar matchgate $H$ with $t = |T|$ external nodes 
such that for any $\gamma \in \{0, 1\}^{|T|}$ of even Hamming weight, $H^\gamma=1$, and 
for any bitstring $\gamma$ of odd Hamming weight, $H^\gamma=0$ (see Section~\ref{Sym-Signatures}).
Clearly $H$ has an even number of nodes, since $H^{00 \ldots 0} = 1$.
We define the planar matchgate $G_1'$ by attaching $H$ to the set $T$ of $G'$
on the side of the semi-circle arc {\it opposite} to all the intersecting
chords in the embedding of $G$. Each node in $T$ is connected
to a distinct external node of $H$ by an edge of weight 1.
We note that
composing $G'$ with $H$ in this fashion does not introduce any more edge crossings, and
 all external nodes $X \cup Y$ still remain on the outer face.

For all $\beta$ where $G_1'^{\beta}$ has an even number of nodes, which
happens exactly when $G^{\beta}$ has an even number of nodes, the following hold:
\begin{equation}
\Gamma_{G_1'}^{\beta} = \sum_{W_{\text{even}}} \Gamma_{G'- W_{\text{even}}}^{\beta}
= \sum_{\text{even }  \gamma} \Gamma_{G'}^{\alpha ( \beta, \gamma)}
= \sum_{\text{even }  \gamma} \Pf(G^{\beta} - W_{\gamma})
= \check{\chi}_G^\beta,
\end{equation}
where
the sum over $W_{\text{even}}$ is  for
all even-sized subsets of $T$,
and $\alpha ( \beta, \gamma)$ is the bit string in 
$\{0, 1\}^{|X \cup T \cup Y|}$ formed by concatenating $\beta$ and $\gamma$, in the proper
order.
If $G_1'^{\beta}$ has an odd
number of nodes, then $\Gamma_{G_1'}^{\beta} = 0$.

There also exists a planar matchgate $H'$  of arity $t = |T|$
such that for any $\beta$ of odd Hamming weight, $H'^\beta=1$,
and for any bitstring $\beta$ of
 even Hamming weight, $H'^\beta=0$ (see Section~\ref{Sym-Signatures}).
Use $H'$ instead of $H$ we can define a planar matchgate  $G_2'$,
which  will have the signature values equal to the
naked character values of $G$ for all $\beta$
for which $G^\beta$ has an odd number of nodes. 
If $G_2'^{\beta}$ has an even
number of nodes, then $\Gamma_{G_2'}^{\beta} = 0$.

This completes the proof.
\end{proof}

Note that a matchcircuit is itself a large general matchgate with only a naked character.
Thus, its character is also expressible as the sum of two signatures
of planar matchgates.

\section{Symmetric Signatures}\label{Sym-Signatures}
We return to planar matchgate signatures.
We say a signature is \emph{even} if it
is the signature of an even matchgate, i.e.,
a matchgate with an even number of nodes.
An even signature has nonzero values only 
for indices of even Hamming weight.
We define an \emph{odd} signature similarly.
A signature $\Gamma$ of a matchgate is \emph{symmetric} if, for all $\alpha,\beta$ of equal Hamming weight, $\Gamma^\alpha=\Gamma^\beta$.
In other words, the value of a signature entry is only a function of how many $1$s are in its index, not their particular pattern.
These signatures are important because they have a clear
combinatorial meaning.
We write a symmetric arity-$k$ signature in the following form $[z_0,z_1,\ldots,z_k]$, where $z_i$ is the value of the signature for
an index of Hamming weight $i$.
The symmetric signatures that obey the MGI have a very concise description, which we prove next.
\begin{theorem}\label{thm:symsig}
If $[z_0,\ldots,z_k]$ is an even symmetric matchgate signature,
then
$z_i=0$ for all odd $i$, and there exist $r_1$ and $r_2$ not both zero 
such that 
for all even $i\geq2$:
\begin{equation}
r_1z_{i-2}=r_2z_i.
\end{equation}
Conversely,
every sequence of values satisfying these conditions is
an even symmetric matchgate signature.
The statement for odd symmetric signatures
is analogous.

Stated equivalently, a sequence is a symmetric matchgate signature iff
it takes the following form:
Alternate entries
of $[z_0,\ldots,z_k]$  are zero 
and the entries at the other alternate positions
form a geometric progression.
\end{theorem}
\begin{proof}
By the Parity Condition, all odd parity entries of the signature
of an even matchgate are zero.
Consider any \emph{even} $i$ and $j$, where $0\le i < j\leq k$.
We invoke the MGI for $\alpha=1^i 1 0^{k-i},\beta=1^{i}01^{j-i-1}0^{k-j}$.
We use the exponentiation notation here to denote repetition.
$\alpha$ has an odd Hamming weight $i+1$ and
 $\beta$ has an odd Hamming weight $j-1$.
Note that $i$ and $j$ being both even implies that $j-i-1 \ge 1$.
Using the fact that $\Gamma$ is symmetric, the MGI under $\alpha,\beta$ becomes:
\begin{equation}\label{eq:symsigstep1}
z_iz_{j}=\sum(\pm)z_{i+2}z_{j-2}.
\end{equation}
There are an odd number ($j-i-1 \ge 1$) of terms in the sum, 
and the terms alternate their signs and begin with a $+$, so we conclude
that
\begin{equation}
z_iz_{j}=z_{i+2}z_{j-2}.
\end{equation}
In particular, if $i$ is even and $0 \le i \le k-4$, then
\begin{equation}
z_iz_{i+4}=z_{i+2}^2.
\end{equation}
If $z_{i+2} \not =0$, then both $z_i \not =0$ and $z_{i+4} \not =0$.
This means that if any even indexed entry that is not the first
or the last even indexed entry (call it a non-extremal entry)
is nonzero, then all even indexed entries are nonzero.
In this case, the geometric progression is established,
with common ratio $z_{i+2}/z_i = z_{i+4}/z_{i+2}$, for even $0 \le i \le k-4$.

Suppose all non-extremal even indexed entries are zero.
If $k \le 3$ then the theorem is self-evident. Suppose $k \ge 4$.
Let $k^* \le k$ be the maximum even index.
Then $k^* \ge 4$ and we have
\begin{equation}
z_0z_{k^*}=z_{2}z_{k^* -2}.
\end{equation}
Note that $k^* -2 \ge 2$ and therefore it is non-extremal.
It follows that $z_0z_{k^*} = 0$ and therefore at most
one extremal even indexed entry can be nonzero.
It is also easy to verify that a sequence satisfying the conditions
of this theorem also satisfies MGI, and hence is a matchgate
signature. (See Section~\ref{Symmetric-construction} for a direct construction.)
This completes the proof for even signatures. The proof for
odd signatures is similar.
The theorem follows.
\end{proof}


Explicitly, there are just four cases for symmetric signatures 
of arity $k$:
\begin{enumerate}
\item $[a^kb^0,0,a^{k-1}b,0,a^{k-2}b^2,0,\ldots,a^0b^k]$
\item $[a^kb^0,0,a^{k-1}b,0,a^{k-2}b^2,0,\ldots,a^0b^k,0]$
\item $[0,a^kb^0,0,a^{k-1}b,0,a^{k-2}b^2,0,\ldots,a^0b^k]$
\item $[0,a^kb^0,0,a^{k-1}b,0,a^{k-2}b^2,0,\ldots,a^0b^k,0]$.
\end{enumerate}


\subsection{Matchgates for Symmetric Signatures}\label{Symmetric-construction}
We have already demonstrated how to build a planar matchgate realizing any MGI-satisfying signature, through a planarizing procedure.
Up till now the only known construction of a matchgate realizing an arbitrary
symmetric signature is through this general procedure.
This is unsatisfactory, since they ought to have more symmetry.
However it is difficult to imagine a geometric construction that is
planar and symmetric for all pairs of external nodes $1 \le i < j \le k$,
if $k \ge 4$.  Now we will present a simple and direct construction.
The constructed matchgates are not geometrically 
symmetric for all pairs of external nodes, but functionally they are,
in terms of the signatures.


We present two closely related matchgate
constructions, one for even symmetric signatures, and the other for odd, which is
a simple modification for the even signature case.
Our constructions for both these cases work regardless if the signature has odd or even arity.

In Fig.\ \ref{fig:symsig6} we have an example of a planar matchgate for an even, arity-6 signature.
Its design can be described as a cycle of triangles which share vertices (each triangle has two weight $x$ edges, and a weight $y$ base).
For odd signatures, the construction is changed very slightly, as shown in Fig.\ \ref{fig:oddsymsig6}.
The only modification is to delete one external node in a matchgate for an even symmetric signature
of arity one higher.

More specifically, to construct in general an even matchgate $G$ of arity $k$,
we first take $k$ triangles with vertices $\{a_i, b_i, c_i\}$ ($1 \le i \le k$).
The edges $\{a_i, b_i\}$ and $\{a_i, c_i\}$ have weight $x$,
and $\{b_i, c_i\}$ has weight $y$.
We link them in a cycle, identifying $c_i$ with $b_{i+1}$,
where the index is counted modulo $k$.
The matchgate $G$ has $k$ external nodes $\{a_1, \ldots, a_k\}$,
and a total of $2k$ nodes.

Consider any $\alpha \in \{0, 1\}^k$ of even Hamming weight.
$\alpha_i =0$ iff $a_i$ remains in $G^\alpha$.
If $\alpha =1^k$, then $G^\alpha$ is a cycle of length $k$.
If $k$ is odd, of course $G^\alpha$ has no perfect matchings.
If $k$ is even, there are exactly two perfect matchings,
each having weight $y^{k/2}$.

Now assume  $\alpha  \not = 1^k$. Then cyclically
$\alpha$ alternates between consecutive  0's (called a {\it 0-run})
and consecutive 1's (called a {\it 1-run}).
Each $a_i$ that remains  in $G^\alpha$  must be matched to either $b_i$ (we call it
{\it left-match}) or $c_i = b_{i+1}$ (we call it {\it right-match}),
both with weight $x$.
Consider any 0-run.
It is clear that either all  $a_i$ within this 0-run
left-match or all right-match.
Next consider a 1-run of $m$ 1's; it is between two 0-runs.
If $m$ is even, then the path  of $m$ edges all with weight $y$
forces the two neighboring 0-runs to take
either both left-match or both right-match.
Moreover, both possibilities are realizable, and in each
case the 1-run contributes a weight $y^{m/2}$.
If $m$ is odd, then  the path  of $m$ edges
forces the two neighboring 0-runs to take
opposite types of left-match and right-match.
Again both possibilities are realizable;
in one case the 1-run contributes a weight $y^{(m-1)/2}$,
and in another case it contributes a weight $y^{(m+1)/2}$. Furthermore, for two
1-runs $1^m$ and $1^{m'}$ both  of odd length and are consecutive in the sense that the only
1-runs in between are of even length,  they contribute a combined weight 
$y^{(m+m')/2}$.
Since $\alpha$ has an even Hamming weight $|\alpha|$, there is
an even number of 1-runs of odd length.
Hence together the 1-runs contribute a weight $y^{|\alpha|/2}$.
There are exactly two perfect matchings in $G^\alpha$, each uniquely determined by
the left-match or right-match choice
 of any particular $a_i$ in $G^\alpha$.
It follows that the signature value is $\Gamma^\alpha= 2 x^{k-|\alpha|}
y^{|\alpha|/2}$.
Clearly by choosing $x$ and $y$ suitably, we can realize
an arbitrary even symmetric signature.  

The construction  for odd symmetric signatures
is to remove one external vertex in the matchgate for
an even symmetric signature of arity one higher.  By 
the general form of odd symmetric signatures, being a sub-signature
$[z_1,  \ldots, z_n]$
of an even symmetric signature 
$[z_0, z_1,  \ldots, z_n]$, the proof is complete.

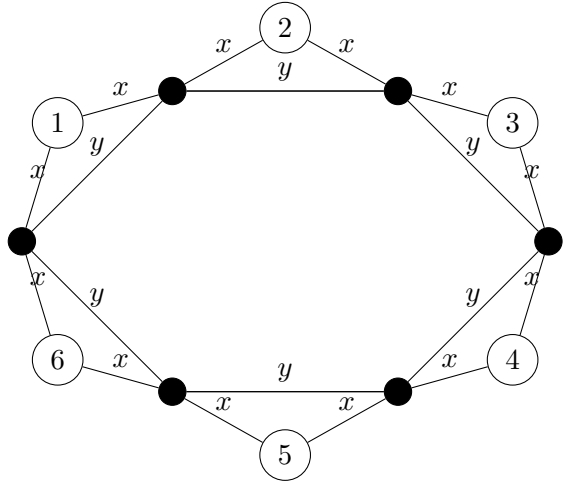
\begin{figure}
\centering
\begin{tikzpicture}
[vertex/.style={circle,draw,fill=black},input/.style={circle,draw}]
\node[vertex](1) at (0,2/1) {};
\node[vertex](2) at (2/1,4/1) {};
\node[vertex](3) at (5/1,4/1) {};
\node[vertex](4) at (7/1,2/1) {};
\node[vertex](5) at (5/1,0) {};
\node[vertex](6) at (2/1,0) {};

\node[input](in1)[above right=1.2cm and 0.1cm of 1] {1};
\node[input](in3)[above left=1.2cm and 0.1cm of 4] {3};

\node[input](in4)[below left=1.2cm and 0.1cm of 4] {4};
\node[input](in6)[below right=1.2cm and 0.1cm of 1] {6};

\path[-] (2) edge node[above=0.5cm,circle,draw](in2) {2} (3);
\path[-] (5) edge node[below=0.5cm,circle,draw](in5) {5} (6);

\foreach \a/\b in {1/2,2/3,3/4,4/5,5/6,6/1} {
	\path[-] (\a) edge node[above] {$y$} (\b);
}
\foreach \a in {1,2,3,4,5,6} {
	\path[-] (in\a) edge node[above] {$x$} (\a);
}
\foreach \a/\b in {1/2,2/3,3/4,4/5,5/6,6/1} {
	\path[-] (in\a) edge node[above] {$x$} (\b);
}
\end{tikzpicture}
\caption{A matchgate for an even, symmetric, arity-$6$ signature.}\label{fig:symsig6}
\end{figure}

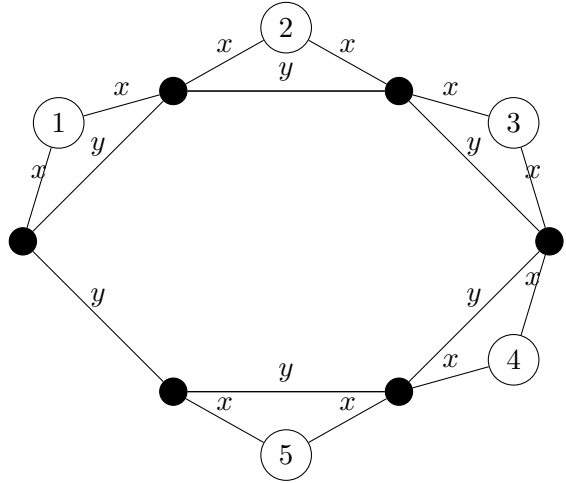
\begin{figure}
\centering
\begin{tikzpicture}
[vertex/.style={circle,draw,fill=black},input/.style={circle,draw}]
\node[vertex](1) at (0,2/1) {};
\node[vertex](2) at (2/1,4/1) {};
\node[vertex](3) at (5/1,4/1) {};
\node[vertex](4) at (7/1,2/1) {};
\node[vertex](5) at (5/1,0) {};
\node[vertex](6) at (2/1,0) {};

\node[input](in1)[above right=1.2cm and 0.1cm of 1] {1};
\node[input](in3)[above left=1.2cm and 0.1cm of 4] {3};

\node[input](in4)[below left=1.2cm and 0.1cm of 4] {4};

\path[-] (2) edge node[above=0.5cm,circle,draw](in2) {2} (3);
\path[-] (5) edge node[below=0.5cm,circle,draw](in5) {5} (6);

\foreach \a/\b in {1/2,2/3,3/4,4/5,5/6,6/1} {
	\path[-] (\a) edge node[above] {$y$} (\b);
}
\foreach \a in {1,2,3,4,5} {
	\path[-] (in\a) edge node[above] {$x$} (\a);
}
\foreach \a/\b in {1/2,2/3,3/4,4/5,5/6} {
	\path[-] (in\a) edge node[above] {$x$} (\b);
}
\end{tikzpicture}
\caption{A matchgate for an odd, symmetric, arity-$5$ signature.}\label{fig:oddsymsig6}
\end{figure}

%

\section{Conclusion}\label{Conclusion}
Substantial work has been built on
top
of MGI in  the signature theory of
matchgates~\cite{cailu11,
cailu07,
lixia08,
cailu09,
landsbergmortonnorine09,
kowalczyk10,
cailu11a,
morton2011,
cailuxia10,
guowilliams12}.
In particular, a number of complexity dichotomy theorems
have been proved that use this understanding of what matchgates
can and cannot compute.
A general theme of these theorems asserts that
a wide class of locally constrained counting problems 
can be classified into three types: (1) Those that are 
computable in polynomial time for
general graphs; (2) Those that are \#P-hard for general graphs
but computable in polynomial time over planar graphs;
and (3) Those that remain  \#P-hard for 
planar graphs.  Moreover type (2) occurs precise
for problems which can be described by signatures that
are realizable by planar matchgates after a holographic
transformation.  This theme is generally proved for
symmetric signatures~\cite{kowalczyk10,cailuxia10,guowilliams12}.
\nocite{Guo-Lu-Valiant}
\nocite{Guo-Huang-Lu-Xia}
For not-necessarily-symmetric signatures, these are only proved in
special cases~\cite{caikowalczykwilliams12}.
This paper provides a firm foundation for this theory 
and for future explorations.

\nocite{caiguowilliams13,landsbergmortonnorine09,morton2011,lixia08}

\bibliography{counting}
\bibliographystyle{plain}
\end{document}